\documentclass[journal]{IEEEtran}
\usepackage{cite}
\usepackage{array}
\usepackage{changepage}
\ifCLASSINFOpdf
   \usepackage[pdftex]{graphicx}
\else
\fi
\usepackage{epstopdf}
\usepackage{bm}
\usepackage{graphicx}
\usepackage{subfigure}
\usepackage{multirow}
\usepackage{multicol}

\usepackage[cmex10]{amsmath}
\usepackage{amsthm}
\newtheorem{Assumption}{Assumption}
\newtheorem{lemma}{Lemma}
\newtheorem{theorem}{Theorem}

\newtheorem{definition}{Definition}
\newtheorem{corollary}{Corollary}
\usepackage{algorithm} 
\usepackage{algorithmicx} 
\usepackage{algpseudocode}

\usepackage{enumerate}

\usepackage{changepage}

\usepackage{mdwmath}
\usepackage{multirow}
\usepackage{booktabs}
\usepackage{dcolumn}
\usepackage{amssymb}

\usepackage{url}

\usepackage{subfloat}
\usepackage{color}

\usepackage{flushend}
\usepackage{fancyhdr}

\ifCLASSINFOpdf
\else
\fi

\begin{document}
%
\title{
Efficient ADMM Decoder for Non-binary LDPC
Codes with Codeword-Independent Performance}
%
%
%

\author{\IEEEauthorblockN{Xiaomeng Guo, \ \ Yongchao Wang}\\
\IEEEauthorblockA{\textit{State Key Laboratory of Intergrated Services Networks} \\
\textit{Xidian University, Xi'an, China}\\
xmguo@stu.xidian.edu.cn, ychwang@mail.xidian.edu.cn}
}

\markboth{Journal of \LaTeX\ Class Files,~Vol.~14, No.~8, August~2015}%
{Shell \MakeLowercase{\textit{et al.}}:
Improved ADMM Decoder for non-binary LDPC codes in $\mathbb{F}_{2^q}$}

\maketitle

\begin{abstract}

In this paper, we devote to devise a non-binary low-density parity-check (LDPC) decoder in Galois fields of characteristic two ($\mathbb{F}_{2^q}$) via the alternating direction method of multipliers (ADMM) technique.
Through the proposed bit embedding technique and the decomposition technique of the three-variables parity-check equation, an efficient ADMM decoding algorithm for non-binary LDPC codes is proposed.
The computation complexity in each ADMM iteration is roughly $\mathcal{O}(nq)$, which is significantly lower than the existing LDPC decoders.
Moreover, we prove that the proposed decoder satisfies the favorable property of the codeword-independent.
Simulation results demonstrate the outstanding performance of the proposed decoder in contrast with state-of-the-art LDPC decoders.
\end{abstract}

\begin{IEEEkeywords}
Non-binary Low-density Parity-check (LDPC) codes, Galois Fields of Characteristic Two, Bit Embedding, Alternating Direction Method of Multipliers (ADMM), Quadratic Programming (QP).
\end{IEEEkeywords}

\IEEEpeerreviewmaketitle

\section{Introduction}
\IEEEPARstart{S}{ince} the high-data-rate requirement in the 5G \cite{5G-survey} wireless communication system, non-binary low-density parity-check (LDPC) codes \cite{Davey-nonBP} in Galois fields of characteristic 2 ($\mathbb{F}_{2^q}$) are more favorable than binary codes in bandwidth-efficient schemes when integrated with higher order modulation \cite{combined-highmodu}.
Moreover, non-binary LDPC codes are competent to eliminate short cycles \cite{better-performance} and resist burst errors \cite{burst-error}, hence they can achieve superior error correction performance than binary codes.

Recently, linear programming (LP) is applied to decode binary LDPC codes by Feldman \emph{et al}. \cite{FeldmanLP} for the first time. In \cite{FeldmanLP}, the objective is designed as the same to maximum likelihood (ML) decoding problem.
 Since then, many researchers pay a lot of attention to the LP decoding techniques and the analyzable decoding guarantees, i.e., convergence and codeword-independent
property.
In contrast with the classical belief propagation (BP) decoders, the general LP solving algorithms still face a difficult challenge, computational complexity, especially when general solvers, such as interior point \cite{interior-point} method or simplex method \cite{revised-simplex}, is used.
Then many works aiming to reduce the decoding complexity have emerged. In \cite{Barman-ADMM}, Barman \emph{et al}. firstly introduced the alternating direction method of multipliers (ADMM) \cite{ADMM} to the LP decoding problem \cite{FeldmanLP}.
The complexity of this ADMM-based decoder is reduced sharply and comparable to the BP decoder, however the time-consuming projection operations onto parity polytopes in each iteration are still included.
In \cite{Bai-wcl}, Bai \emph{et al}. utilized the three-variables parity-check equations decomposition technique to construct an efficient ADMM-based decoder.
 Besides the efforts on reducing the complexity, many works focus on the error correction performance in low SNR regions.
 Taagavi and Siegel in \cite{adaptive-LP} employed requisite parity-check constraints adaptively to the proposed LP decoder to improve the decoding performance.
 To eliminate the undesired pseudo-codewords, a series of cut-generating algorithms \cite{cut-plane-algorithm}-\cite{Adaptive-cut} were proposed to improved the error correction performance of the LP decoders.
 Moreover, Liu \emph{et al}. in \cite{penalty-decoder} and Bai \emph{et al}. in \cite{bai-admm-qp-binary} constructed an ADMM-based variable node penalized decoder, Wei \emph{et al}. in \cite{CheckNodePenalized} proposed an ADMM-based check node penalized decoder.
 These decoders all added a non-convex term to the LP decoding objective, which promoted the error correction performance of LP decoding in low SNR region and did not suffer from the error-floor like BP decoders.

 LP techniques were first introduced to non-binary LDPC codes by Flanagan in \cite{Flanagan}.
 As the same challenge the binary LP decoding encountered when used general LP solvers, the non-binary LP decoding also faces the complexity problem.
 To resolve this problem, a coordinate ascent method was applied in \cite{ILP-nonbinary} and \cite{LCLP-nonbinary} to solve the dual problem of the original LP problem \cite{ILP-binary} and \cite{LCLP-binary}.
 In \cite{Rosnes-ALP}, Rosnes proposed adaptive linear programming decoding of linear codes over prime fields and dynamic programming was used in the separation algorithm of the decoding polytopes described by the resulting inequalities.
 Besides, Honda \emph{et al}. in \cite{Fast-LP-nonbinary} proposed a constant-weight embedding which used a binary vector to represent the element in $\mathbb{F}_{2^q}$, however no further efficient algorithm was presented to tackle the resulting LP problem.
 Furthermore, Liu and Draper in \cite{Liu-nonbinary-journal} proposed  (non)penalized LP decoding algorithms for non-binary LDPC codes in $\mathbb{F}_{2^q}$ based on ADMM technique.
 Especially, the parity-check constraints are equivalently transformed to a series of binary constraints through a factor graph representation.
 Then the ADMM algorithm was applied to solve the resulting decoding problem.
 However, the time-consuming Euclidean projections onto high dimensional simplex polytopes are still required in each iteration of the proposed ADMM decoder.
 Lately, Wang and Bai in \cite{wang-nonbinary} proposed a proximal-ADMM decoder for non-binary LDPC codes.
 The three variables parity-check equations decomposition technique was implemented together with the same equivalent transform of the constraints in \cite{Liu-nonbinary-journal}.
 Then the proximal-ADMM algorithm was used to generate a block-parallel updating scheme.
 To date, there are two techniques, named as Flanagan embedding and Constant-Weight embedding \cite{Flanagan}\cite{Liu-nonbinary-journal}, which can map element in $\mathbb{F}_{2^q}$ one-to-one correspondent to a length $2^q-1$ or length $2^q$ binary vector in real space. Either results in computational complexity of the corresponding decoding algorithms  is proportional to $2^q$, i.e., the size of the considered Galois fields \cite{Liu-nonbinary-journal}\cite{wang-nonbinary}.

  In this paper, our main concern is devising a novel non-binary LDPC decoder with competitive error correction performance and lower decoding complexity, meanwhile the codeword-independent performance has to be theoretically-guaranteed. The main contents of this paper are summarized as follows:
 \begin{itemize}
  \item Different from the embedding techniques in \cite{Flanagan} and  \cite{Fast-LP-nonbinary}, we propose a new embedding method, named bit embedding, to represent the elements in finite fields by a binary vector of dimension $q$ in real space, under which the decoding complexity grows linearly instead of grows exponentially (\cite{Liu-nonbinary-journal}\cite{wang-nonbinary}) as the considered length $q$ grows.
  We also introduce a equivalent transform of the non-binary parity-check constraints to a series of binary parity-check constraints under the bit embedding.
  \item Based on the three variables parity-check equations decomposition technique and its equivalent binary inequalities formulation, we decompose the multi-variables binary parity-check constraints and transform the non-binary ML decoding problem to an equivalent linear integer program.
  Then we relax the binary constraints to box constraints and add a quadratic penalty term to the objective, a new quadratic programming (QP) decoding model is built for non-binary LDPC codes in $\mathbb{F}_{2^q}$.
  \item We exploit the ADMM algorithm to solve the resulting QP decoding model.
   Utilizing the special intrinsic structure of the QP problem, all the variables in one ADMM iteration can be updated in a full-parallel pattern.
  \item The performance of the proposed ADMM decoding algorithm satisfies the favorable property of the codeword-independent. Moreover, its complexity in each iteration is roughly $\mathcal{O}(nq)$.
  \item We show through numerical simulations that the proposed ADMM decoder outperforms state-of-the-art decoders in error correction performance and complexity.
  \end{itemize}
  To facilitate reading of this paper, notations are explained explicitly in Table \ref{notation-table}. The operator ``$\otimes$'' are specifically defined as follows.
      \begin{itemize}
        \item For matrices $\mathbf{A} \in \mathbb{R}^{m\times n} $ and $\mathbf{B} \in \mathbb{R}^{p\times q}$, $\mathbf{A} \otimes \mathbf{B}=\begin{bmatrix}
    a_{11}\mathbf{B} & \cdots & a_{1n}\mathbf{B} \\
       \vdots        & \ddots & \vdots \\
    a_{m1}\mathbf{B} & \cdots & a_{mn}\mathbf{B}
  \end{bmatrix}
  \in \mathbb{R}^{mp \times nq}$.
   \item For column vectors $\mathbf{a}=[a_1,\ldots,a_m]^T$ and $\mathbf{b}=(b_1,\ldots,b_n)^T$, $\mathbf{a} \otimes \mathbf{b}=[a_1\mathbf{b};\ldots;a_m\mathbf{b}]\in \mathbb{R}^{mn}$.
  \item For column vector $\mathbf{a}=[a_1,\ldots,a_m]^T$ and matrix $\mathbf{B} \in \mathbb{R}^{p\times q}$, $\mathbf{a} \otimes \mathbf{B}=[a_1\mathbf{B},\ldots,a_m\mathbf{B}]\in \mathbb{R}^{mp\times n}$.
  \end{itemize}
  \begin{table}[t]
\caption{Notations and Descriptions.}
\label{notation-table}
\renewcommand{\arraystretch}{1.3}
\begin{center}
\begin{tabular}{|c|l|}
\hline
\hline
\textbf{Notations}  &  \hspace{3cm} \textbf{Descriptions}                                     \\\hline
$\mathbb{F}_{2^q}$  & Galois field of characteristic two                    \\\hline
$\mathbb{R}$        & The set of real numbers                                            \\\hline
$\mathbf{A}$        & Matrix                                            \\\hline
$\mathbf{a}$        & Column vector
 \\\hline
$a$                 & Scalar
 \\\hline
 $\{0,1\}^a$        & a-length binary column vector
 \\\hline
 \hspace{0.2cm}$\{0,1\}^{a\times b}$         & a-by-b binary matrix
 \\\hline
$[\mathbf{a}; \mathbf{b}]$ or $[\mathbf{A}; \mathbf{B}]$             & Vectors or matrices are concatenated in column-wise                 \\\hline
$[\mathbf{a}, \mathbf{b}]$ or $[\mathbf{A},\mathbf{B}]$             & Vectors or matrices are concatenated in row-wise                      \\\hline
$\mathbf{1}_{a}$ or $\mathbf{1}_{a\times b}$             & Length-$a$ all-ones vector or $a$-by-$b$ all-ones matrix                   \\\hline
$\mathbf{0}_{a}$ or $\mathbf{0}_{a\times b}$             & Length-$a$ all-zeros vector or $a$-by-$b$ all-zeros matrix                      \\\hline
$\mathbf{I}_{a}$             & $a\times a$ identity matrix                   \\\hline
$(\cdot)^T$         & Transpose operator                                    \\\hline
$\|\cdot\|_{2}$     &  $\ell_2$-norm                                  \\\hline
$\preceq$           & Generalized inequality          \\\hline
$\otimes$           & Kronecker product                                     \\\hline
$\delta_{\mathbf{A}}$             & Spectral norm of matrix $\mathbf{A}$                      \\\hline
$\lambda_{\min}(\mathbf{A}^T\!\mathbf{A})$    & Minimum eigenvalue of matrix $\mathbf{A}^T\mathbf{A}$                                 \\\hline
$\underset{\mathcal{X}}\Pi$    & Euclidean projection onto  set $\mathcal{X}$                                 \\\hline
\hline
\end{tabular}
\end{center}
\end{table}

The rest of this paper is organized as follows.
In Section \ref{ML-QP}, we transform the non-binary ML decoding problem to an equivalent linear integer program based on the decomposition technique and the proposed bit embedding technique.
In Section \ref{ADMM-algorithm}, we establish a relaxed QP decoding problem for non-binary LDPC codes. Then an efficient ADMM decoding algorithm is proposed to solve the resulting relaxed QP problem.
Section \ref{performance-analysis} presents the codeword-independent and the complexity analysis of the proposed ADMM decoding algorithm.
Simulation results show the effectiveness of our proposed ADMM decoder in Section \ref{simulation}.
Section \ref{conclusion} concludes this paper.

\section{ML Decoding problem formulation and its equivalent lineal integer program}\label{ML-QP}
Consider a non-binary LDPC codeword defined by an $m$-by-$n$ parity-check matrix $\mathbf{H}$ in $\mathbb{F}_{2^q}$. Its feasible codeword set $ \mathcal{C} $ can be denoted by
\begin{equation}\label{feasible_set}
\mathcal{C}=\{\mathbf{c}|\mathbf{h}^T_j\mathbf{c}=0,j\in \mathcal{J},\mathbf{c}\in\mathbb{F}_{2^q}^n\},
\end{equation}
where $\mathbf{h}^T_j$, $j\in \mathcal{J}=\{1,2,\cdots,m\}$ denotes the $j$-th row vector of the parity-check matrix $\mathbf{H}$.

Pass the codeword $\mathbf{c}\in \mathcal{C}$ through an additional white Gaussian noise (AWGN) channel and denote its output as $\mathbf{r}$. In the receiver, the aim of ML decoding is to determine which codeword has the largest a priori probability $p(\mathbf{r}|\mathbf{c})$ throughout the codeword set $\mathcal{C}$. So the ML decoding problem can be described as
\begin{subequations}\label{ori_ML_c}
        \begin{align}
                \underset{\mathbf{c}}{\rm  max} \  &p(\mathbf{r}|\mathbf{c}),
                 \label{ori_ML_c_a} \\
                 {\rm s.t.}    \ & \mathbf{h}^T_j\mathbf{c}=0,\ j\in \mathcal{J},\mathbf{c}\in\mathbb{F}_{2^q}^n. \label{ori_ML_c_b}
        \end{align}
    \end{subequations}
For any integer $c_i\in\mathbf{c}$, we introduce a new mapping that embeds elements in $\mathbb{F}_{2^q}$ into the Euclidean space of dimension $q$.

\begin{definition}\label{def-bit-embed}
{\textit{Bit embedding}:} Let $f:\mathbb{F}_{2^q}\mapsto\{0,1\}^q$ be a mapping such that for $c_i\in\mathbb{F}_{2^q}$,
\begin{equation*}
f(c_i):=\mathbf{x}_i=[x_{i,1},\cdots,x_{i,\sigma},\cdots,x_{i,q}]^T,\footnote{Here we show some examples of bit embedding: $2\in\mathbb{F}_{2^2}\mapsto[1,0]^T$, $7\in\mathbb{F}_{2^3}\mapsto[1,1,1]^T$, $15\in\mathbb{F}_{2^4}\mapsto[1,1,1,0]^T$.}
\end{equation*}
where $x_{i,\sigma}$ comes from the polynomial representation of $c_i\in\mathbb{F}_{2^q}$
\begin{equation}\label{embed_rule}
p(2)=\sum^{q}_{\sigma=1}2^{\sigma-1}x_{i,\sigma}.
\end{equation}
\end{definition}

Under this embedding rule, any codeword $\mathbf{c}$ can be mapped to a binary vector $\mathbf{x}=[\mathbf{x}_1;\cdots;\mathbf{x}_n]\in\{0,1\}^{nq}$.
We name the new embedding technique as the bit embedding, since the embedding vector $\mathbf{x}_i$ is the bit representation of $c_i\in \mathbb{F}_{2^q}$.
We call $\mathbf{x}$ as an equivalent binary codeword of the non-binary codeword $\mathbf{c}$.

 Define a new vector $\hat{\mathbf{h}}^{(j,i)}=h_{j,i}\cdot [2^{q-1},\cdots,2,1]$ with $q$-length, where all elements in $\hat{\mathbf{h}}^{(j,i)}$ and the multiplication are in $\mathbb{F}_{2^q}$.
 Map all the elements in $\hat{\mathbf{h}}^{(j,i)}$ to a column vector by the bit embedding and concatenate all the column vectors in row-wise, which leads to the embedding result $\hat{\mathbf{H}}_{j,i}\in\{ 0,1 \}^{q\times q}$.\footnote{We present an example of the above procedure. Consider $h_{j,i}=6\in\mathbb{F}_{2^3}$, then $6\cdot[4, 2, 1]=[5, 7, 6]$. Embed $5$, $7$, $6$ respectively, $5\mapsto[1,0,1]^T$, $7\mapsto[1,1,1]^T$, $6\mapsto[1,1,0]^T$. Then $\hat{\mathbf{H}}_{j,i}=\begin{bmatrix}1 \ 1 \ 1\\ 0\ 1\ 1\\ 1\ 1\ 0\end{bmatrix}$.} Then, we have the following lemma.

 \begin{lemma}\label{equivalent_codeword}
 With $\mathbf{x}_i$ as the equivalent binary codeword of the non-binary codeword $c_i$ according to the mapping rule \eqref{embed_rule}. Then, $\hat{\mathbf{H}}_{j,i} \mathbf{x}_i$ is an equivalent binary codeword to $h_{j,i}c_i$ mapped according to the same rule.
 \end{lemma}
\begin{proof}
 See Appendix \ref{proof of equivalent_codeword}.
\end{proof}
Based on Lemma \ref{equivalent_codeword}, the parity-check equation $h_{j,1}c_1+h_{j,2}c_2+\cdots+h_{j,n}c_n=0$ in \eqref{ori_ML_c_b} is equivalent to the following parity-check equations in $\mathbb{F}_2$:
\begin{equation}
\hat{\mathbf{H}}_{j,1} \mathbf{x}_1+\cdots+\hat{\mathbf{H}}_{j,n} \mathbf{x}_n=0,
\end{equation}
which is equal to\label{new-parity-check-equations}
\begin{equation}\label{new-parity-check-equations}
\hat{\mathbf{H}}_{j} \mathbf{x}=0,
\end{equation}
where $\hat{\mathbf{H}}_{j}=[\hat{\mathbf{H}}_{j,1},\cdots,\hat{\mathbf{H}}_{j,n}]\in\{0,1\}^{q\times nq}$.

\begin{figure*}[tp]
  \centering
  \includegraphics[width=17.2cm,height=6.3cm]{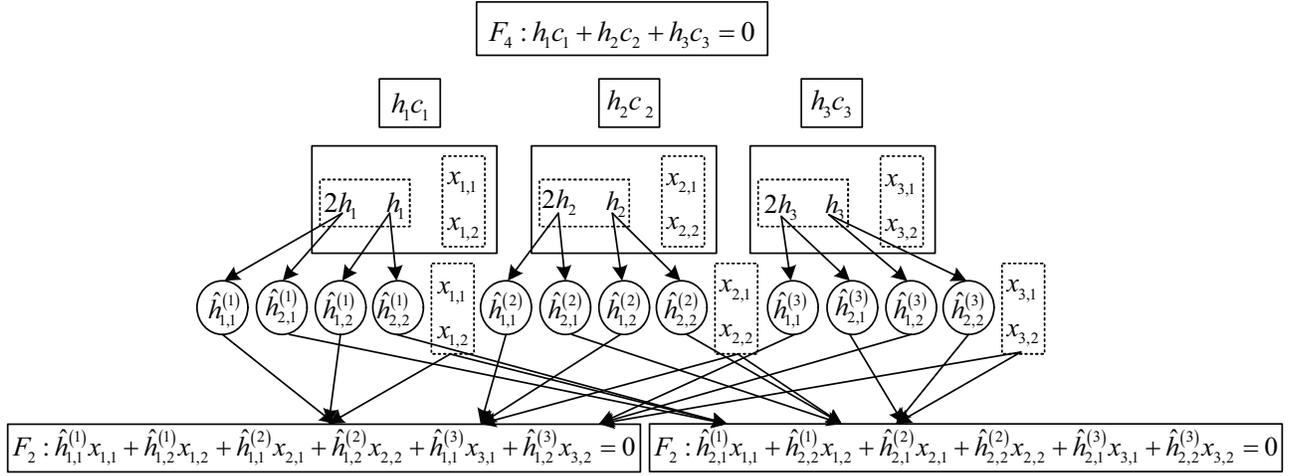}
  \caption{The graph representation for equivalence between the three-variables parity-check equation in $\mathbb{F}_4$ and the two six-variables parity-check equations in $\mathbb{F}_2$, where $h_1, h_2, h_3\in\mathbb{F}_4\setminus0$, $c_1, c_2, c_3\in\mathbb{F}_4$ and $\hat{h}_{1,1}^{(1)},\hat{h}_{1,2}^{(1)},\hat{h}_{2,1}^{(1)},\hat{h}_{2,2}^{(1)},\hat{h}_{1,1}^{(2)},\hat{h}_{1,2}^{(2)},\hat{h}_{2,1}^{(2)},\hat{h}_{2,2}^{(2)},\hat{h}_{1,1}^{(3)},\hat{h}_{1,2}^{(3)},\hat{h}_{2,1}^{(3)},\hat{h}_{2,2}^{(3)},x_{1,1},x_{1,2},x_{2,1},x_{2,2},x_{3,1},x_{3,2}\in\mathbb{F}_2$.}
  \label{factor_graph}
\end{figure*}
Assume the degree of $\mathbf{h}_j^T$ in \eqref{ori_ML_c_b} is three, then Figure \ref{factor_graph} shows an example of equivalence three-variables parity-check equation in $\mathbb{F}_4$ and two six-variables parity-check equations in $\mathbb{F}_2$ following Lemma \ref{equivalent_codeword} .

 Let $\mathcal{X}$ denotes the set consisting of all equivalent binary codewords. Then the ML decoding problem \eqref{ori_ML_c} can be transformed to
\begin{subequations}\label{ori_ML_x}
        \begin{align}
                \underset{\mathbf{x}}{\rm  max} \  &p(\mathbf{r}|\mathbf{x}),
                 \label{ori_ML_x_a} \\
                 {\rm s.t.}    \ & \hat{\mathbf{H}}_{j}\mathbf{x}=0,\ j\in\mathcal{J}, \label{ori_ML_x_b}\\
                 &\mathbf{x}\in\{0,1\}^{nq}.
        \end{align}
    \end{subequations}

  \begin{Assumption}
  The output of the channel is bit-independent, i.e., $p(\mathbf{r}|\mathbf{x})=\displaystyle\prod^{n}_{i=1}\displaystyle\prod_{\sigma=1}^{q}p(r_i|x_{i,\sigma})$.
  \end{Assumption}
  Based on the above assumption, we can derive the optimization model \eqref{ori_ML_x_a} as
  \begin{equation}\label{pre_loglike}
  	\begin{aligned}
  	\underset{\mathbf{x}\in\mathcal{X}}{\rm argmax}\ p(\mathbf{r}|\mathbf{x})=&\underset{\mathbf{x}\in\mathcal{X}}{\rm argmax}\displaystyle\prod^{n}_{i=1}\displaystyle\prod_{\sigma=1}^{q}p(r_i|x_{i,\sigma})\\
  	=&\underset{\mathbf{x}\in\mathcal{X}}{\rm argmin}\displaystyle\sum^n_{i=1}\displaystyle\sum^{q}_{\sigma=1}-{\rm {log}} \ p(r_i|x_{i,\sigma}).
  	\end{aligned}
  \end{equation}
  Putting the constant $\displaystyle\sum^n_{i=1}\displaystyle\sum^{q}_{\sigma
  =1}{\rm {log} }\ p(r_i|x_{i,\sigma}=0)$ into \eqref{pre_loglike}, we get

   \begin{equation}\label{loglike}
  	\begin{aligned}
  	\underset{\mathbf{x}\in\mathcal{X}}{\rm argmax}\ p(\mathbf{r}|\mathbf{x})=&\underset{\mathbf{x}\in\mathcal{X}}{\rm argmin}\displaystyle\sum^n_{i=1}\displaystyle\sum^{q}_{\sigma=1}{\rm{log}}\ \frac{p(r_i|x_{i,\sigma}=0)}{p(r_i|x_{i,\sigma})} \\
  	=& \underset{\mathbf{x}\in\mathcal{X}}{\rm argmin}\displaystyle\sum^n_{i=1}\displaystyle\sum^{q}_{\sigma=1}x_{i,t}{\rm{log}}\ \frac{p(r_i|x_{i,\sigma}=0)}{p(r_i|x_{i,\sigma}=1)} \\
  	=&\underset{\mathbf{x}\in\mathcal{X}}{\rm argmin}\ {\pmb\gamma}^T\mathbf{x},
  	\end{aligned}
  \end{equation}
  where $\pmb\gamma\in\mathbb{R}^{nq}$ is defined by
  \begin{equation}\label{loglike-analysis}
	\begin{aligned}
  	\pmb\gamma=&\left[{\rm{log}} \frac{p(r_1|x_{1,1}=0)}{p(r_1|x_{1,1}=1)},\cdots,{\rm{log}} \frac{p(r_1|x_{1,q}=0)}{p(r_1|x_{1,q}=1)},\right.\\
  	&\left. \cdots,{\rm{log}} \frac{p(r_n|x_{n,1}=0)}{p(r_n|x_{n,1}=1)},\cdots,{\rm{log}} \frac{p(r_n|x_{n,q}=0)}{p(r_n|x_{n,q}=1)}\right].
  	\end{aligned}
  \end{equation}
  Then the ML decoding problem \eqref{ori_ML_x} can be formulated as the following integer program
 \begin{subequations}\label{ML_x}
        \begin{align}
                \underset{\mathbf{x}}{\rm  min} \  &{\pmb\gamma}^T\mathbf{x},
                 \label{ML_x_a} \\
                 {\rm s.t.}    \ & \hat{\mathbf{H}}_{j}\mathbf{x}=0,\ j\in\mathcal{J}, \label{ML_x_b}\\
                 &\mathbf{x}\in\{0,1\}^{nq}\label{ML_x_c}.
        \end{align}
    \end{subequations}

Observing problem \eqref{ML_x}, one can see that the difficulty of solving it lies in how to handle the non-convex binary parity-check constraints \eqref{ML_x_b} and the integer constraint \eqref{ML_x_c}. To address these challenges, we first decompose the parity-check equations in \eqref{ML_x_b} to a series of three-variables binary ones and express them by some linear inequalities.

First we consider the three-variables parity-check equation in $\mathbb{F}_2$
\begin{equation}\label{three-parity-check}
x_1+x_2+x_3=0,x_i\in\{0,1\},i\in\{1,2,3\},
\end{equation}
which can be equivalent to the following inequality system defined in real space
\begin{equation}\label{four_inequalities}
	\begin{aligned}
  	&x_1\leq x_2+x_3,\quad x_2\leq x_1+x_3,\\
  	&x_3\leq x_1+x_2,\quad x_1+x_2+x_3\leq 2,\\
  	&x_1,x_2,x_3\in\{0,1\}
  	\end{aligned}
  \end{equation}
in the sense that \eqref{three-parity-check} and \eqref{four_inequalities} have the same solutions. Define
\begin{equation}\label{w T matrix}
   \begin{aligned}
   \mathbf{w}=\begin{bmatrix}\ 0\ \\ \ 0\ \\ \ 0\ \\ \ 2\ \end{bmatrix},\mathbf{T} = \begin{bmatrix}
                           ~~1  &-1  & -1 \\
                          -1  &~~1 & -1 \\
                          -1  &-1  &~~1 \\
                          ~~1  &~~1 &~~1
                       \end{bmatrix}.
   \end{aligned}
   \end{equation}
Then, \eqref{three-parity-check} can be rewritten as
\begin{equation}
\mathbf{Tx}\preceq \mathbf{w},\mathbf{x}\in\{0,1\}^3,
\end{equation}
where $\mathbf{x}=[x_1,x_2,x_3]^T$.

This technique can be applied to the general parity-check equation with more than three variables and formulate its equivalent expression.
Especially, we consider the $t$-th ($t\in\{1,\cdots,mq\}$) parity-check equation in \eqref{ML_x_b}.
Without loss of generality, we assume it involves $d_t\geq 3$ variables, which are denoted by $x_{\sigma_{1}},\cdots,x_{\sigma_{d_t}}$, the corresponding coefficients $h_{\sigma_1},\cdots,h_{\sigma_{d_t}}$ are 1 and others are 0.
For the first two variables $x_{\sigma_1}$ and $x_{\sigma_2}$, we introduce an auxiliary variable $a_1$ and let them satisfy the three-variables parity-check equation $a_1+h_{\sigma_1}x_{\sigma_1}+h_{\sigma_2}x_{\sigma_2}=0$. For the variables $x_{\sigma_3},\cdots,x_{\sigma_{d_t-3}}$, we introduce two auxiliary variables $a_{k-1}$ and $a_{k}$ for one variable and let them satisfy $a_{k-1}+h_{\sigma_{k+1}}x_{\sigma_{k+1}}+a_{k}=0, k\in\{2,\cdots,d_t-3\}$. Let the auxiliary variable $a_{d_t-3}$ and the last two variables satisfy $a_{d_t-3}+h_{\sigma_{d_t-1}}x_{\sigma_{d_t-1}}+h_{\sigma_{d_t}}x_{\sigma_{d_t}}=0$. Following this procedure, we can avoid the special kind of three-variables parity-check equations that consists of all auxiliaries variables.

Consequently, the multi-variables parity-check equations are decomposed
into finite three-variables ones. Moreover, the numbers of the introduced auxiliary variables and the corresponding three-variables parity-check equations are calculated by $d_t-3$ and $d_t-2$ respectively. Then, the total numbers of the three-variables parity-check equations and the introduced auxiliary variables corresponding to the constraints \eqref{ML_x_b} are
\begin{equation}\label{aux_num}
\Gamma_c=\displaystyle\sum_{t=1}^{mq}(d_t-2),\ \Gamma_a=\displaystyle\sum_{t=1}^{mq}(d_t-3),
\end{equation}
respectively.
Define $\mathbf{v}=[\mathbf{x};\mathbf{u}]\in\{0,1\}^{nq+\Gamma_a}$, where auxiliary variables $\mathbf{u}\in\{0,1\}^{\Gamma_a}$. Define a variable-selecting matrix $\mathbf{Q}_\tau\in\{0,1\}^{3\times(nq+\Gamma_a)}$ of which very row vector includes only one "1" and corresponds to the variables in the $\tau$-th three-variables parity-check equation, $\tau=1,\cdots,\Gamma_c$. Then it can be expressed as $\mathbf{T}\mathbf{Q}_\tau\mathbf{v}\preceq \mathbf{w}$. Furthermore, we define
\begin{subequations}\label{q A b}
        \begin{align}
                \mathbf{q}&=[{\pmb\gamma};\mathbf{0}_{\Gamma_a}]^T\in\mathbb{R}^{N},\label{q} \\
                \mathbf{A}&=[\mathbf{T}\mathbf{Q}_1;\cdots;\mathbf{T}\mathbf{Q}_\tau;\cdots;\mathbf{T}\mathbf{Q}_{\Gamma_c}]\in\mathbb{R}^{M\times N},\label{A}\\
                 \mathbf{b}&=\mathbf{1}_{\Gamma_c}\otimes\mathbf{w}\in\mathbb{R}^{M},\label{b}
        \end{align}
    \end{subequations}
 where $M=4\Gamma_c$ and $N=nq+\Gamma_a$. Then the ML decoding problem \eqref{ML_x} is equivalent to the following linear integer program
 \begin{subequations}\label{LP}
        \begin{align}
                 \min_{\mathbf{v}}\ &\mathbf{q}^T\mathbf{v},\label{LP-a} \\
                {\rm s.t.}\ &\mathbf{Av}\preceq\mathbf{b},\label{LP-b}\\
                 &\mathbf{v}\in\{0,1\}^{nq+\Gamma_a},\label{LP-c}.
        \end{align}

    \end{subequations}

 Due to the discrete binary constraints \eqref{LP-c}, problem \eqref{LP} is non-convex and difficult to solve. In the following, we relax it to a continuous, non-convex, but tractable model.

\section{Non-convex QP formulation and its ADMM solving algorithm}\label{ADMM-algorithm}
\subsection{Non-Convex QP Formulation}
A classical way for the binary constraint \eqref{LP-c} is to relax it to a box constraint $\mathbf{v}\in[0,1]^{nq+\Gamma_a}$, which can simplify the non-convex problem \eqref{LP} to a convex one. However, the resulting optimization problem's solution could be a fractional especially when the decoder works in the low SNR regions. Therefore we add a quadratic penalty term into the objective \eqref{LP-a},
which leads to the following non-convex QP decoder
 \begin{subequations}\label{LP-relax}
        \begin{align}
                 \min_{\mathbf{v}}\ &\mathbf{q}^T\mathbf{v}-\frac{\alpha}{2}\|\mathbf{v}-0.5\|^2_2,\label{LP_relax-a} \\
                {\rm s.t.}\ &\mathbf{Av}\preceq\mathbf{b},\label{LP-relax-b}\\
                 &\mathbf{v}\in[0,1]^{nq+\Gamma_a}\label{LP-relax-c}.
        \end{align}

    \end{subequations}
Inspired by the intrinsic structure of model \eqref{LP-relax}, we propose an efficient solving algorithm via the ADMM technique. In it, the variables in every ADMM iteration can be solved analytically and in a full-parallel pattern. Meanwhile, the computational complexity in each iteration scales linearly with the non-binary LDPC codes length and the considered length $q$.

\subsection{ADMM Algorithm Framework}
By introducing auxiliary variable $\mathbf{e}_1$ and $\mathbf{e}_2$, the decoding problem \eqref{LP-relax} can be equivalent to
\begin{subequations}\label{LP-relax_aux}
        \begin{align}
                 \min_{\mathbf{v}}\ &\mathbf{q}^T\mathbf{v}-\frac{\alpha}{2}\|\mathbf{v}-0.5\|^2_2,\label{LP_relax_aux-a} \\
                {\rm s.t.}\ &\mathbf{Av}+\mathbf{e}_1=\mathbf{b},\label{LP-relax-aux-b}\\
                 & \mathbf{v}=\mathbf{e}_2,\label{LP-relax-aux-c}\\
                 &\mathbf{e}_1\succeq \mathbf{0}_{4\Gamma_c},\mathbf{e}_2\in[0,1]^{nq+\Gamma_a}.\label{LP-relax-aux-d}
        \end{align}

    \end{subequations}
Its augment Lagrangian function can be written as
{\small{
\begin{equation}\label{aug-Lagran}
	\begin{aligned}
  	\mathcal{L}_\mu(\mathbf{v},\mathbf{e}_1,\mathbf{e}_2,\mathbf{y}_1,\mathbf{y}_2)&=\mathbf{q}^T\mathbf{v}\!-\!\frac{\alpha}{2}\|\mathbf{v}-0.5\|^2_2\!+\!\mathbf{y}_1^T(\mathbf{Av}+\mathbf{e}_1-\mathbf{b})\\
  	&\!+\!\mathbf{y}_2^T(\mathbf{v}\!-\!\mathbf{e}_2)\!+\!\frac{\mu}{2}\|\!\mathbf{Av}\!+\!\mathbf{e}_1\!-\!\mathbf{b}\!\|^2_2\!+\!\frac{\mu}{2}\|\!\mathbf{v}\!-\!\mathbf{e}_2\!\|^2_2,
  	\end{aligned}
  \end{equation}
  }}where $\mathbf{y}_1\in \mathbb{R}^{4\Gamma_c}$ and $\mathbf{y}_2\in\mathbb{R}^{nq+\Gamma_a}$ are Lagrangian multipliers and $\mu>0$ is a preset penalty parameter. Based on the augment Lagrangian \eqref{aug-Lagran}, the ADMM solving algorithm for model \eqref{LP-relax_aux} can be described as
    \begin{subequations}\label{subproblems}
        \begin{align}
        \mathbf{v}^{k+1}&=\underset{\mathbf{v}}{{{\rm argmin}}} \ \mathcal{L}_{\mu}\ (\mathbf{v},\mathbf{e}_1^{k},\mathbf{e}_2^{k},\mathbf{y}_1^{k},\mathbf{y}_2^{k}),\label{subproblems-a}\\
        \mathbf{e}_1^{k+1}&=\underset{\mathbf{e}_1\succeq 0}{{\rm argmin}} \ \mathcal{L}_{\mu}\ (\mathbf{v}^{k+1},\mathbf{e}_1,\mathbf{e}_2^k,\mathbf{y}_1^k,\mathbf{y}_2^k),\label{subproblems-b}\\
                 \mathbf{e}_2^{k+1}&=\underset{0\preceq\mathbf{e}_2\preceq 1}{{\rm argmin}} \ \mathcal{L}_{\mu} \label{subproblems-d}\ (\mathbf{v}^{k+1},\mathbf{e}_1^{k+1},\mathbf{e}_2,\mathbf{y}_1^k,\mathbf{y}_2^k),\\
                \mathbf{y}_1^{k+1}&=\mathbf{y}_1^k+\mu(\mathbf{Av}^{k+1}+\mathbf{e}_1^{k+1}-\mathbf{b}),\label{subproblems-c}\\
                 \mathbf{y}_2^{k+1}&=\mathbf{y}_2^{k}+\mu(\mathbf{v}^{k+1}-\mathbf{e}_2^{k+1}),\label{subproblems-e}
        \end{align}
    \end{subequations}where $k$ denotes the iteration number.

 In the following we show that the subproblems \eqref{subproblems-a}-\eqref{subproblems-d} can be solved efficiently be exploiting the inherent structure of the model.
 \subsubsection{Solving Subproblem \eqref{subproblems-a}}
Based on the fact that matrix $\mathbf{A}$ has the property of orthogonality in columns\footnote{Since $\mathbf{T}$ is column-wise orthogonal matrix, $\mathbf{A}  $ is also orthogonal in columns. Besides, it can be found that elements in $\mathbf{A}$ are either 0,1 or -1.}, hence $\mathbf{A}^T\mathbf{A}$ is a diagonal matrix, which means that variables in problem \eqref{subproblems-a} are separable. Therefore the $nq+\Gamma_a$ entries of $\mathbf{v}^{k+1}$ can be obtained by solving the following $nq+\Gamma_a$ subproblems in parallel
\begin{equation}\label{sub-a}
\min_{v_i} \frac{1}{2}\!(\!\mu s_i-\!\alpha+\!\mu\!)v_i^2\!+\!(\!q_i+y_{2,i}^k+\!\mu e_{2,i}^{k+1}\!\!+\!\frac{1}{2}\!\alpha\!+\mathbf{a}^T_i\!(\!\mathbf{y}_1^k\!+\!\mu(\mathbf{e}_1^{k+1}\!\!-\!\mathbf{b})))v_i,
\end{equation}
 where $\mathbf{s}={\rm diag}(\mathbf{A}^T\mathbf{A})=[s_1,\cdots,s_{nq+\Gamma_a}]^T$ and $\mathbf{a}^T_i $ denotes the $i$th row vector of matrix $\mathbf{A}$.

 To guarantee that \eqref{sub-a} is strong convex, we choose proper $\alpha$ and $\mu$ and let them satisfy $\mu (s_{\rm min}+1)>\alpha$, where $s_{\rm min}$ denote the minimum value of the elements in vector $\mathbf{s}$. Then the solution of problem \eqref{sub-a} can be obtained by setting the gradient of objective \eqref{sub-a} to zero, i.e.,
 \begin{equation}\label{v-update-parallel}
 v_i^{k+1}=\left(\frac{\mathbf{a}_i^T(\mathbf{b}-\mathbf{e}_1^k-\frac{\mathbf{y}_1^k}{\mu})+e_{2,i}^k-\frac{y_{2,i}^k}{\mu}-\phi_i}{\theta_i}\right)
 \end{equation}
 where $i\in\{1,\cdots,nq+\Gamma_a\}$, $\phi_i=\frac{2q_i+\alpha}{2\mu}$ and $\theta_i=s_i-\frac{\alpha}{\mu}+1$.

 \subsubsection{Solving Subproblem \eqref{subproblems-b}}
Obviously, all the entries in $\mathbf{e}_1$ are separable in both objective and constraints. Thus the optimal solution of problem \eqref{subproblems-b} can be obtained by setting the gradient of the function $\mathcal{L}_{\mu}\ (\mathbf{v}^{k+1},\mathbf{e}_1,\mathbf{e}_2^k,\mathbf{y}_1^k,\mathbf{y}_2^k)$ to zero, then project the resulting solution of the corresponding linear equation to region $[0,+\infty]^{4\Gamma_c}$, i.e.,
\begin{equation}
 \mathbf{e}_1^{k+1}=\underset{[0,+\infty]^{4\Gamma_c}}{\Pi}\left((\mathbf{b}-\mathbf{A}\mathbf{v}^{k+1})-\frac{\mathbf{y}_1^k}{\mu}\right)
 \end{equation}
 Explicitly, all the entries of $\mathbf{e}_1$ can also be obtained in parallel by
 \begin{equation}\label{e_update_parallel}
 e^{k+1}_{1,j}=\underset{[0,+\infty]}{\Pi}\left((b_j-\mathbf{a}^T_j\mathbf{v}^{k+1})-\frac{y^k_{1,j}}{\mu}\right)
 \end{equation}
 where $j\in\{1,\cdots,4\Gamma_c\}$ and $\underset{[0,+\infty]}{\Pi}(\cdot)$ denotes the projection operator onto the positive quadrant $[0,+\infty]$.

 \subsubsection{Solving Subproblem \eqref{subproblems-d}}
To obtain the solution of $\mathbf{e}_2$, similarly, we set the gradient of the function $\mathcal{L}_{\mu}\ (\mathbf{v}^{k+1},\mathbf{e}_1^{k+1},\mathbf{e}_2,\mathbf{y}_1^k,\mathbf{y}_2^k)$ to zero, then project the resulting solution of the corresponding linear equation to region $[0,1]^{nq+\Gamma_a}$, i.e.,
\begin{equation}
 \mathbf{e}_2^{k+1}=\underset{[0,1]^{nq+\Gamma_a}}{\Pi}\left(\mathbf{v}^{k+1}+\frac{\mathbf{y}_2^{k}}{\mu}\right)
 \end{equation}
 Specially, all the entries of $\mathbf{e}_2$ can also be obtained in parallel by
 \begin{equation}\label{e2_update_parallel}
 e^{k+1}_{2,i}=\underset{[0,1]}{\Pi}\left(v_i^{k+1}+\frac{y_{2,i}^k}{\mu}\right)
 \end{equation}
 where $i\in\{1,\cdots,nq+\Gamma_a\}$ and $\underset{[0,1]}{\Pi}(\cdot)$ denotes the projection operator onto the interval $[0,1]$.

 Since variables $\mathbf{y}_1$ and $\mathbf{y}_2$ are engaged in \eqref{e_update_parallel}  and \eqref{e2_update_parallel} by its scaled form $\frac{\mathbf{y}_1}{\mu}$ and $\frac{\mathbf{y}_2}{\mu}$, thus we update its scaled form to save multiplications. Therefore, we define $\tilde{e}^{k+1}_{1,j}=b_j-\mathbf{a}^T_j\mathbf{v}^{k+1}-\frac{y^k_{1,j}}{\mu}$ and $\tilde{e}^{k+1}_{2,i}=v_i^{k+1}+\frac{y_{2,i}^k}{\mu}$. Then, each element in $\mathbf{y}_1$ and $\mathbf{y}_2$ can be obtained through
 \begin{equation}\label{y1-updata-parallel}
 \frac{y_{1,j}^{k+1}}{\mu}=\left\{\begin{array}{cc}
0, & \text { if } \tilde{e}_{1,j}^{k+1} \geq 0, \\
-\tilde{e}_{1,j}^{k+1}, & \text { otherwise ,}
\end{array}\right.
 \end{equation}

 \begin{equation}\label{y2-updata-parallel}
 \frac{y_{2,i}^{k+1}}{\mu}=\left\{\begin{array}{cc}
0, & \text { if } 1\geq \tilde{e}_{1,j}^{k+1} \geq 0, \\
\tilde{e}_{2,i}^{k+1}, & \text { if }\tilde{e}_{1,j}^{k+1} \leq 0,\\
\tilde{e}_{2,i}^{k+1}-1, & \text { if }\tilde{e}_{1,j}^{k+1} \geq 1.
\end{array}\right.
 \end{equation}

 In Algorithm 1, we summarize the proposed ADMM decoding algorithm for solving model \eqref{LP-relax_aux}.

 \section{Performance analysis}\label{performance-analysis}
\subsection{Codeword-independent property}
In the following, we prove one important characteristic of linear codes, that under the symmetric channel condition, the failure probability of the decoding algorithm is independent of the codeword that is transmitted. We prove this property holds for our proposed non-binary ADMM decoder under the bit embedding.
\begin{theorem}\label{theorem-codeword-independent}
Under the channel symmetry condition in the sense of \cite{Flanagan}, the probability that Algorithm \ref{ADMM-QP} fails is independent of the codeword that was transmitted.
\end{theorem}
\begin{proof}
See Appendix \ref{proof of codeword-independent}.
\end{proof}

\subsection{Computational complexity}
In this subsection, we we focus on analyzing the complexity of the proposed Algorithm 1 in each iteration.
Moreover the presented result only includes multiplications since they cost the most computational resource in practice.
Now Since all the entries in matrix $\mathbf{A}$ are either 0, 1 or -1, we can see that all multiplications with regard to $\mathbf{A}$ can be replaced by addition. Moreover, $\phi_i$ and $\theta_i$ in \eqref{v-update-parallel} can be calculated in advance before we launch the ADMM iterations.

Consider updating $\mathbf{v}^{k+1}$ first. Based on the analytic solution of $v_i^{k+1}$ in \eqref{v-update-parallel}, one can find that updating $\mathbf{v}^{k+1}$ costs $nq+\Gamma_a$ multiplications.
Moreover, observing \eqref{e_update_parallel} and \eqref{e2_update_parallel}, we can find that it takes no multiplication operation to update $\mathbf{e}_1^{k+1}$ and $\mathbf{e}_2^{k+1}$. In addition, observing variables $\mathbf{y}_1$ and $\mathbf{y}_2$ in \eqref{e_update_parallel}, \eqref{e2_update_parallel} and \eqref{v-update-parallel}, apparently their scalded form $\frac{\mathbf{y}_1}{\mu}$ and $\frac{\mathbf{y}_2}{\mu}$ are engaged in updating. Hence we conclude that calculating $\frac{\mathbf{y}_1^{k+1}}{\mu}$ and $\frac{\mathbf{y}_2^{k+1}}{\mu}$ only need some addition operations from \eqref{y1-updata-parallel} and \eqref{y2-updata-parallel}. Based on the above analysis, the total multiplications of Algorithm 1 in each iteration are $nq+\Gamma_a$.

 Moreover, one can find that $\Gamma_a\leq mq(d-3)=nq(1-R)(d-3)$, where $R$ denotes the code rate and $d$ is the largest check node degree. This implies that $\Gamma_a$ is comparable to code length $n$ since $d\ll n$ in the case of LDPC codes. Hence we conclude that the computational complexity of the proposed ADMM decoding algorithm in every iteration scales linearly the non-binary LDPC code length and the considered length $q$, or roughly $\mathcal{O}(nq)$.
Moreover, we should note that the entries in $\mathbf{v}$, $\mathbf{e}_1$, $\mathbf{e}_2$, $\mathbf{y}_1$ and $\mathbf{y}_2$ can be updated in a full-parallel pattern in contrast with \cite{wang-nonbinary}, which indicates it is free of time-consuming Euclidean projection onto high dimensional parity-check/simplex polytopes (\cite{Liu-nonbinary-journal}).
 \begin{algorithm}[t]
 \caption{The proposed ADMM decoding algorithm}
 \label{ADMM-QP}
\begin{algorithmic}[1]
 \State Construct $\hat{\mathbf{H}}_j$ in \eqref{new-parity-check-equations} based on the parity-check matrix $\mathbf{H}$. Compute log-likelihood ration $\pmb\gamma$ via \eqref{loglike-analysis}. Construct $\mathbf{q}$, $\mathbf{A}$ and $\mathbf{b}$ via \eqref{q}, \eqref{A}, \eqref{b} respectively. Let $\mu (s_{\rm min}+1)>\alpha$. Compute $\phi_i$ and $\theta_i$ in \eqref{sub-a}.
 \State Initialize $\{\mathbf{v},\mathbf{e}_1,\mathbf{e}_2,\mathbf{y}_1,\mathbf{y}_2\}$ as the all-zeros vectors. Set $k=0$.
 \State  \textbf{Repeat}
 \State \hspace{0.3cm}  Update $\{v_i^{k+1}|i=1,\cdots,nq+\Gamma_a\}$ in parallel by \eqref{v-update-parallel}.
 \State \hspace{0.3cm} Update $\{e_{1,j}^{k{}+1}|j=1,\cdots,4\Gamma_c\}$ in parallel by \eqref{e_update_parallel}.
  \State \hspace{0.3cm}  Update $\{e_{2,i}^{k{}+1}|i=1,\cdots,nq+\Gamma_a\}$ in parallel by \eqref{e2_update_parallel}.
 \State \hspace{0.3cm}  Update $\frac{\mathbf{y}_1^{k+1}}{\mu}$ and $\frac{\mathbf{y}_2^{k+1}}{\mu}$ in parallel by \eqref{y1-updata-parallel} and \eqref{y2-updata-parallel} respectively.
 \State \hspace{0.3cm}  $k=k+1$.
 \State \textbf{Until} $\|\mathbf{Av}^k+\mathbf{e}_1^k-\mathbf{b}\|^2_2$ or $\|\mathbf{v}^k-\mathbf{e}_2^k\|$ is less than some preset $\epsilon$.
\end{algorithmic}
\end{algorithm}

\section{Simulation results}\label{simulation}
\begin{figure*}[tp]
\subfigure[Tanner (1055,424)-$\mathcal{C}_{1}$ in $\mathbb{F}_4$ with QPSK modulation.]{
    \begin{minipage}{8.5cm}
    \centering
        \includegraphics[width=3.5in,height=2.7in]{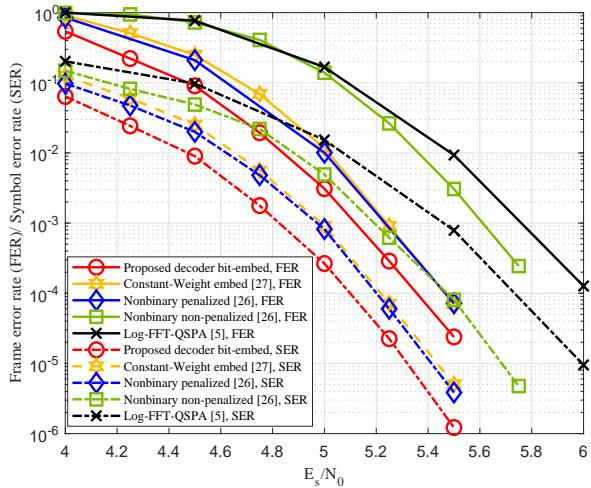}
            \label{fer1055}
    \end{minipage}%
    }
    \subfigure[{PEG (504,252)-$\mathcal{C}_{2}$ in $\mathbb{F}_4$ with QPSK modulation.}]{
    \begin{minipage}{8.5cm}
    \centering
        \includegraphics[width=3.5in,height=2.7in]{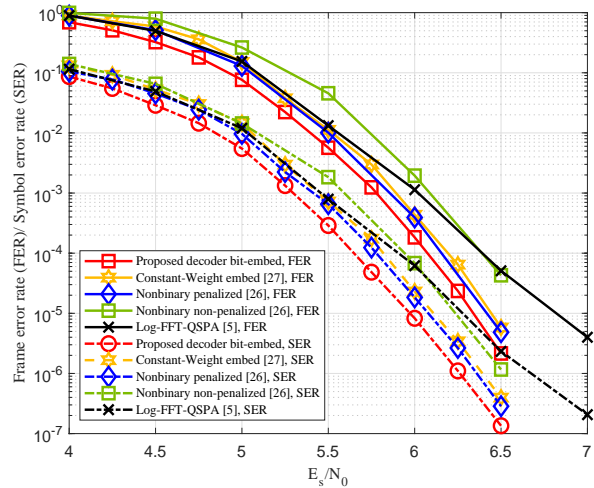}
            \label{fer504}
    \end{minipage}%
    }\
 \subfigure[{MacKay (204,102)-$\mathcal{C}_{3}$ in $\mathbb{F}_8$ with 8PSK modulation.}]{
    \begin{minipage}{8.5cm}
    \centering
        \includegraphics[width=3.5in,height=2.7in]{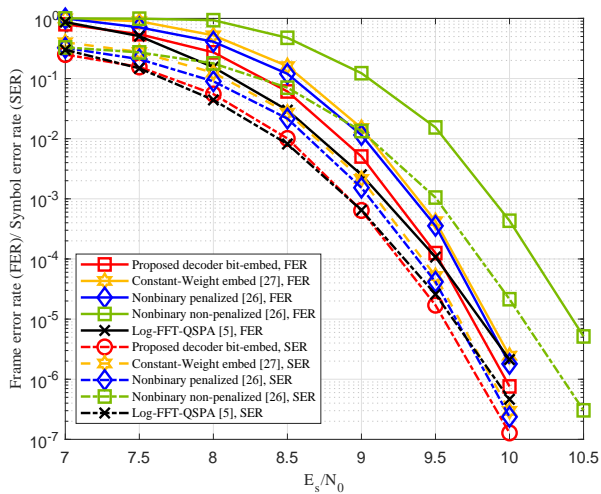}
            \label{fer204}
    \end{minipage}%
    }
    \subfigure[Tanner (155,64)-$\mathcal{C}_{4}$ in $\mathbb{F}_{16}$ with 16QAM.]{
    \begin{minipage}{8.5cm}
    \centering
        \includegraphics[width=3.5in,height=2.7in]{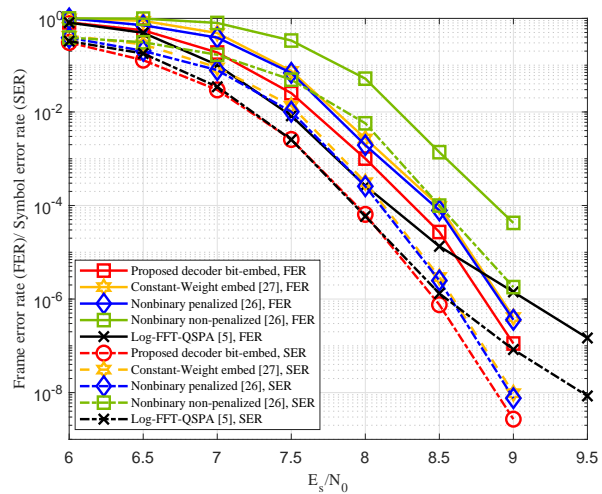}
            \label{fer155}
    \end{minipage}%
    }
     \centering
    \caption{{Comparisons of FER/SER performance for four LDPC codes from \cite{Mackay-code} and \cite{Tanner-code} in different Galois fields  with different modulations.}}
    \label{fer_ser}
 \end{figure*}

\begin{figure*}[tp]
\subfigure[Tanner (1055,424) code $\mathcal{C}_{1}$ in $\mathbb{F}_4$ with QPSK modulation.]{
    \begin{minipage}{8.5cm}
    \centering
        \includegraphics[width=3.5in,height=2.6in]{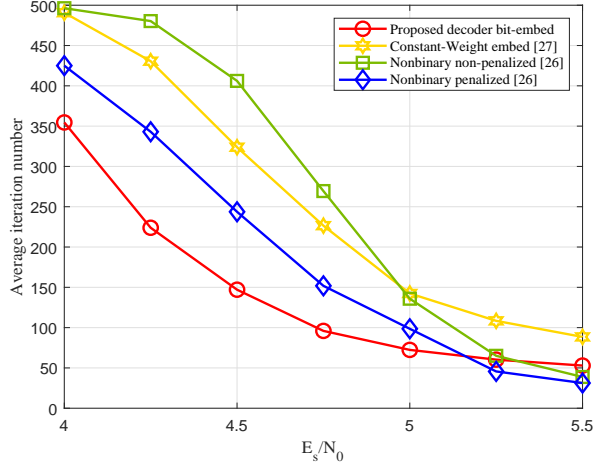}
            \label{iter1055}
    \end{minipage}%
    }
    \subfigure[{PEG (504,252) code $\mathcal{C}_{2}$ in $\mathbb{F}_4$ with QPSK modulation.}]{
    \begin{minipage}{8.5cm}
    \centering
        \includegraphics[width=3.5in,height=2.6in]{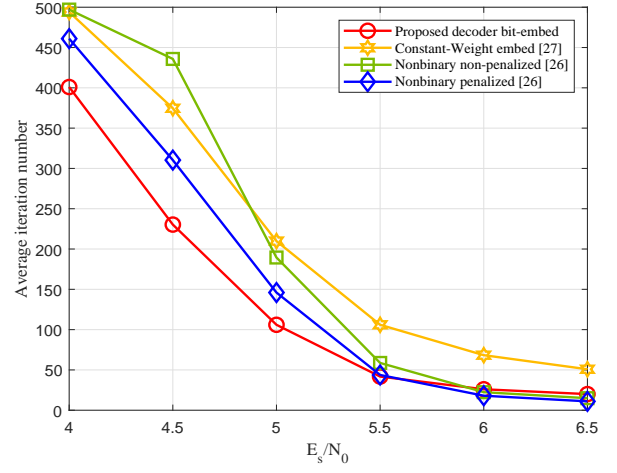}
            \label{iter504}
    \end{minipage}%
    }\
\subfigure[{MacKay (204,102) code $\mathcal{C}_{3}$ in $\mathbb{F}_8$ with 8PSK modulation.}]{
    \begin{minipage}{8.5cm}
    \centering
        \includegraphics[width=3.5in,height=2.6in]{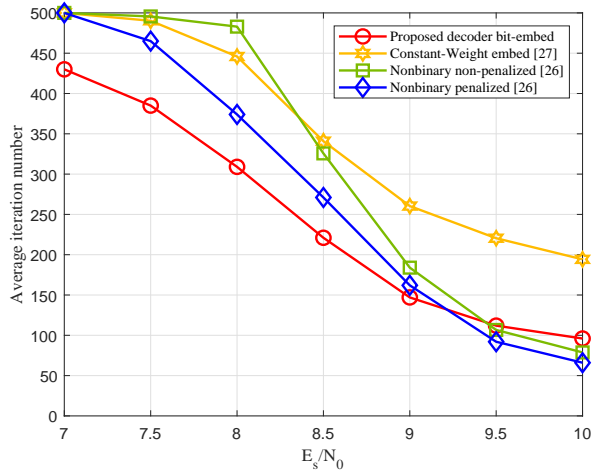}
            \label{iter204}
    \end{minipage}%
    }
    \subfigure[Tanner (155,64) code $\mathcal{C}_{4}$ in $\mathbb{F}_{16}$ with 16QAM modulation.]{
    \begin{minipage}{8.5cm}
    \centering
        \includegraphics[width=3.5in,height=2.6in]{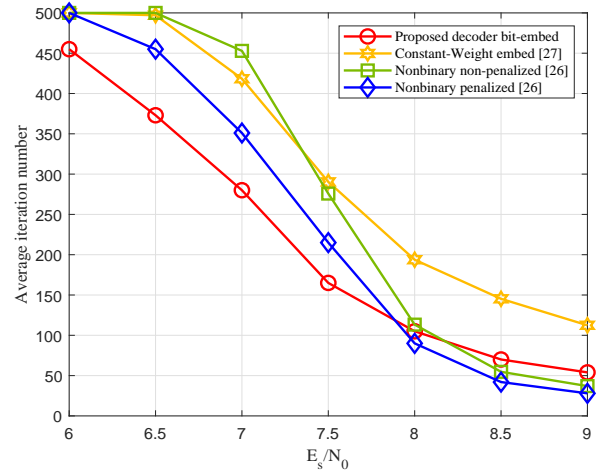}
            \label{iter155}
    \end{minipage}%
    }
     \centering
    \caption{Comparisons of average number of iterations for four LDPC codes from \cite{Mackay-code} and \cite{Tanner-code} using different decoders with different modulations.
    }
    \label{iter4}
 \end{figure*}

\begin{figure*}[tp]
\subfigure[{Tanner (1055,424)-$\mathcal{C}_{1}$ in $\mathbb{F}_4$ with QPSK modulation.}]{
    \begin{minipage}{8.5cm}
    \centering
        \includegraphics[width=3.5in,height=2.6in]{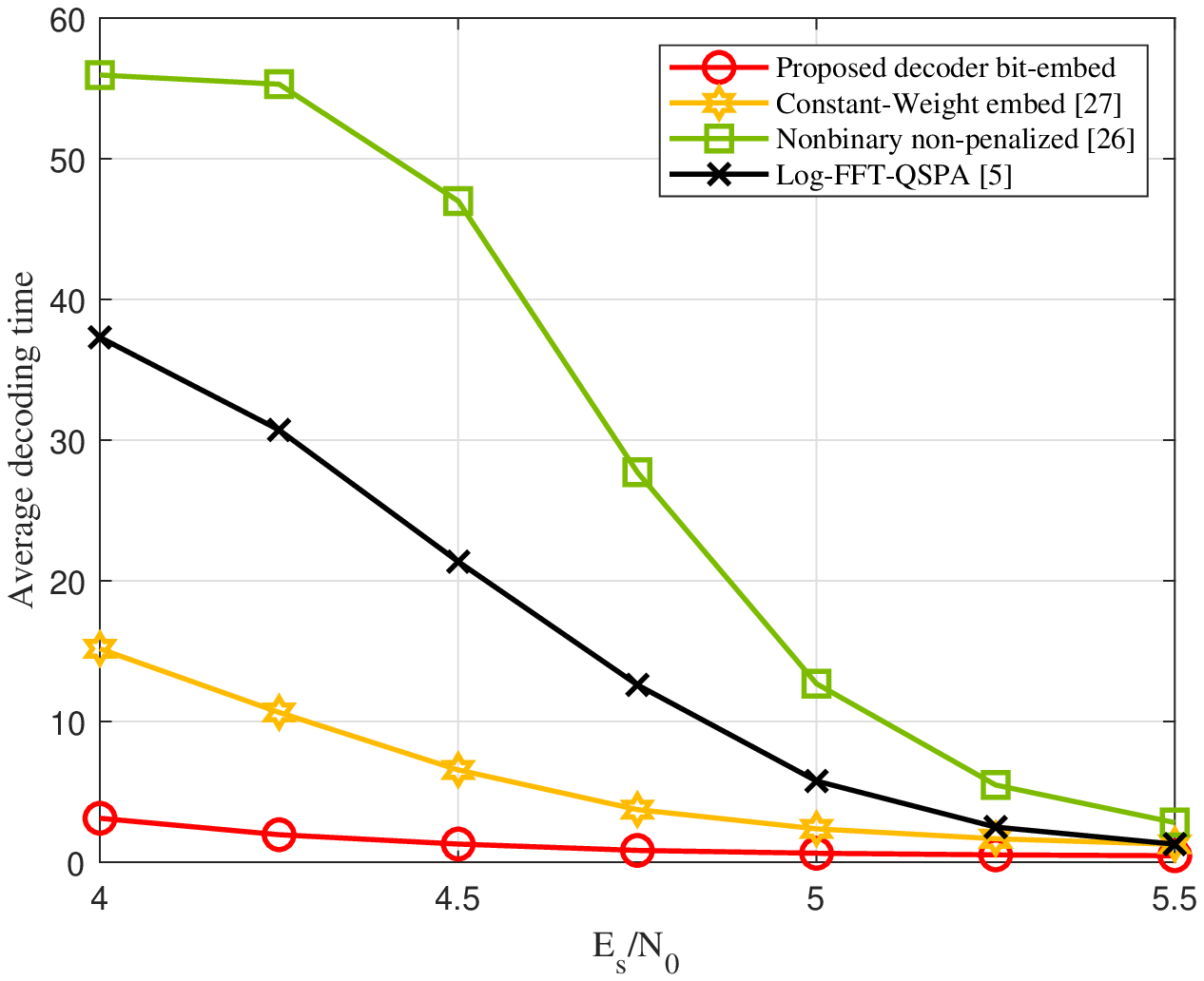}
            \label{time1055}
    \end{minipage}%
    }
    \subfigure[{PEG (504,252)-$\mathcal{C}_{2}$ in $\mathbb{F}_4$ with QPSK modulation.}]{
    \begin{minipage}{8.5cm}
    \centering
        \includegraphics[width=3.5in,height=2.6in]{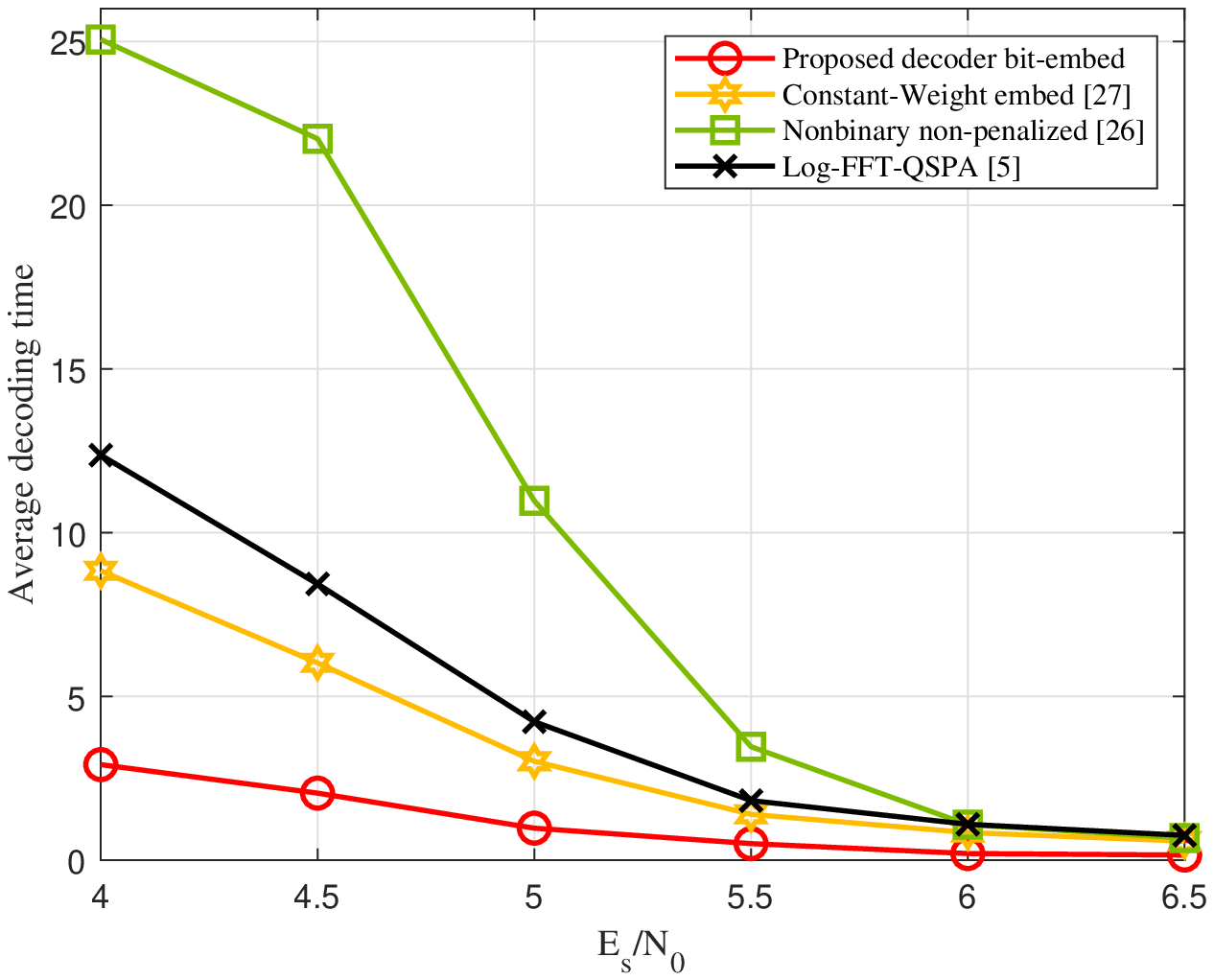}
            \label{time504}
    \end{minipage}%
    }\
\subfigure[{MacKay (204,102)-$\mathcal{C}_{3}$ in $\mathbb{F}_8$ with 8PSK modulation.}]{
    \begin{minipage}{8.5cm}
    \centering
        \includegraphics[width=3.5in,height=2.6in]{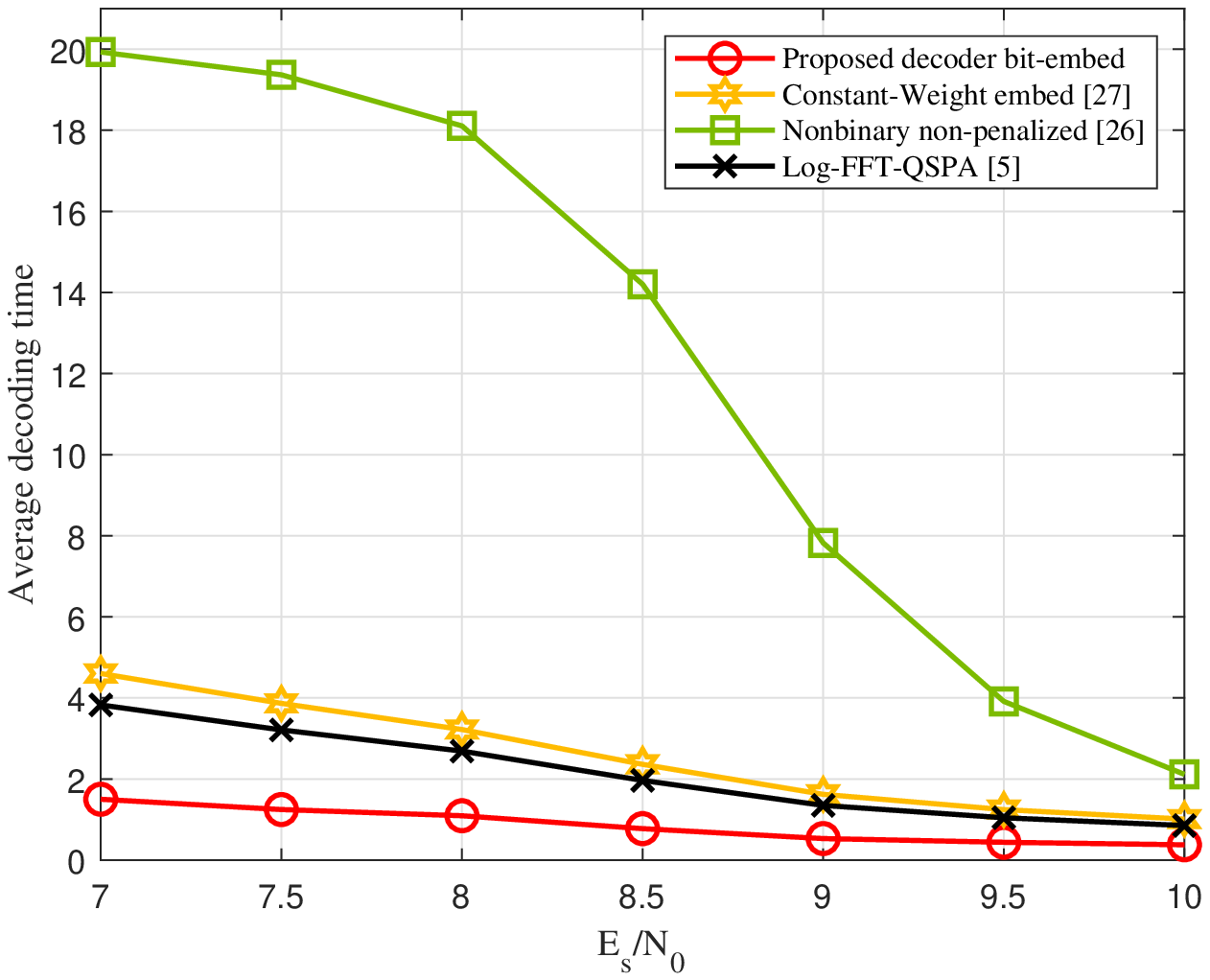}
            \label{time204}
    \end{minipage}%
    }
    \subfigure[{Tanner (155,64)-$\mathcal{C}_{4}$ in $\mathbb{F}_{16}$ with 16QAM modulation.}]{
    \begin{minipage}{8.5cm}
    \centering
        \includegraphics[width=3.5in,height=2.6in]{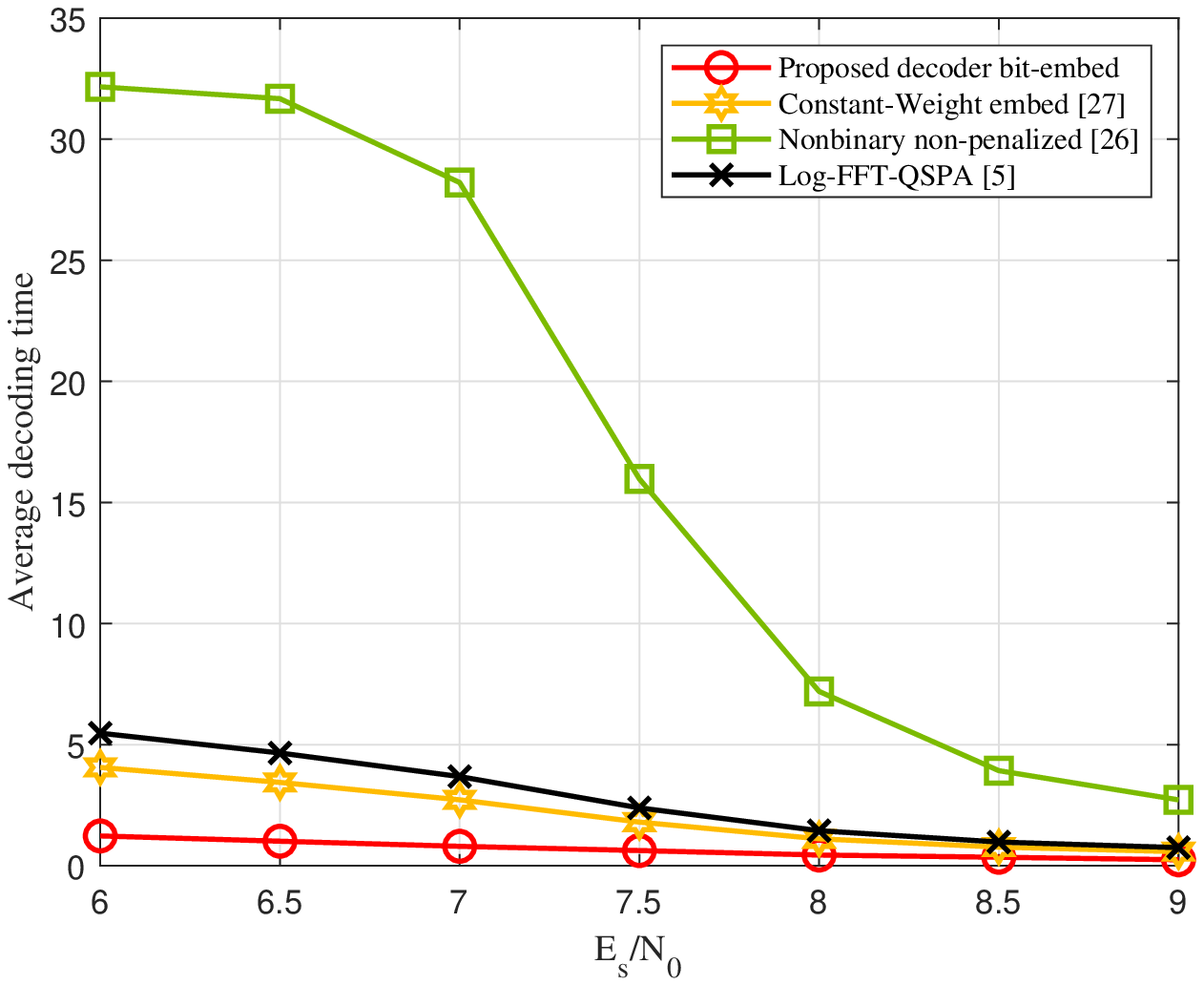}
            \label{iter155}
    \end{minipage}%
    }
     \centering
     \caption{Comparisons of average decoding time for four LDPC codes from \cite{Mackay-code} and \cite{Tanner-code} using different decoders with different modulations.
    }
    \label{time4}
 \end{figure*}

In this section, we provide several numerical results to show the error-correction performance (frame error rate (FER) and symbol error rate (SER)) and the decoding efficiency of the proposed non-binary ADMM decoder, which are also compared with several state-of-the-art non-binary LDPC decoders.

The four codes we considered are Tanner (1055,424)-$\mathcal{C}_1$ from \cite{Tanner-code}, irregular PEG (504,252)-$\mathcal{C}_2$ from \cite{Mackay-code}, rate-0.5 MacKay (204,102)-$\mathcal{C}_3$ from \cite{Mackay-code} and Tanner (155,64)-$\mathcal{C}_4$ from \cite{Tanner-code}. The same parity-check matrix are used for $\mathcal{C}_1$ and $\mathcal{C}_2$, of which all the entries are binary but are settled as elements in $\mathbb{F}_4$. Meanwhile $\mathcal{C}_1$ and $\mathcal{C}_2$ are modulated by quadrature phase shift keying (QPSK). For $\mathcal{C}_3$, the first entry to the sixth entry of each row of the parity-check matrix are set to 1/4/6/5/2/1$\in\mathbb{F}_8$, as the same setting in \cite{LCLP-nonbinary}. The parity-check matrix of $\mathcal{C}_4$ are also binary but all the entries are settled to the elements in $\mathbb{F}_{16}$. 8 phase-shift keying (8PSK) and sixteen quadrature amplitude modulation (16QAM) are used to modulate $\mathcal{C}_3$ and $\mathcal{C}_4$ respectively. The modulated symbols are transmitted over additive white Gaussian noise (AWGN) channel. We compare our proposed ADMM decoder to state-of-the-art decoders, i.e., proximal-ADMM decoder \cite{wang-nonbinary}, logarithm-domain fast-Fourier-transform-based Q-ary sum-product algorithm (Log-FFT-QSPA) \cite{burst-error}, non-binary non-penalized decoder and non-binary penalized decoder in \cite{Liu-nonbinary-journal}. We set the parameters of the Algorithm 1 as follows: parameter $\alpha$ are set to 0.9 for all four codes, parameter $\mu$ are set to $1.1,0.9,0.8,0.7$ for all codes respectively. The randomly generated codewords are used as input. When $\|\mathbf{Av}^{k}+\mathbf{e}_1^k-\mathbf{b}\|^2_2\leq 10^{-5}$ or $\|\mathbf{v}^k-\mathbf{e}_2^k\|^2_2$ is satisfied or the maximum iteration number 500 is reached, we stop the ADMM iteration. The simulation environment is MATLAB 2019a/Windows 10 on a computer with 3GHz Intel i5-9400 CPU and 16GB RAM.

Figure \ref{fer_ser} demonstrates the FER/SER performance of codes $\mathcal{C}_1$-$\mathcal{C}_4$ under different decoders. All the data points are generated based on at least 200 error frames except for the last two points in each curve, which are generated based on 50 error frames due to the limited computational resources. From Figure \ref{fer1055} to \ref{fer155}, we can observe that the proposed decoder has a better performance than the proximal-ADMM decoder in \cite{wang-nonbinary} and the penalized decoder in \cite{Liu-nonbinary-journal}, and are much better than the non-penalized decoder in \cite{Liu-nonbinary-journal}. Especially for the long non-binary LDPC codes in Figure \ref{fer1055}, the proposed decoder outperforms the decoder in \cite{wang-nonbinary} and \cite{Liu-nonbinary-journal} about 0.1dB. Moreover, the proposed decoder shows a better error-correction performance than the Log-FFT-QSPA decoder \cite{burst-error} in Figure \ref{fer1055}, \ref{fer504}. In Figure \ref{fer204} and \ref{fer155}, the proposed decoder has a similar performance to the Log-FFT-QSPA decoder in a low SNR region, but performs better in a high SNR region. In summary, our proposed ADMM decoder exhibits a superior error-correction performance than state-of-the-art decoders.

\begin{table}[t]
\caption{Comparison of the number of constraints and variables in different decoding Algorithm.}
\label{constraints-variable-table}
\renewcommand{\arraystretch}{1.3}
\begin{center}
\begin{tabular}{|c|c|c|c|c|}
\hline
\hline
\multirow{2}*{Codeword}  &  \multicolumn{2}{c|}{CW embedding in \cite{wang-nonbinary}} &\multicolumn{2}{c|}{the proposed decoder}                                     \\ \cline{2-5}
  & constraints &variables   & constraints &variables                \\ \hline
  $\mathcal{C}_1$ & 25109 &9284&15192&4642  \\ \hline
  $\mathcal{C}_2$ & 19387 &6092&12124&3572  \\ \hline
  $\mathcal{C}_3$ & 11934 &4080&8160&2346  \\ \hline
  $\mathcal{C}_4$ & 17081 &5456&4464&1364  \\ \hline
\hline
\end{tabular}
\end{center}
\end{table}
Figure \ref{iter4} and Figure \ref{time4} shows the averaged number of iteration and the averaged decoding time of four codes for the proposed ADMM decoder, the proximal-ADMM decoder \cite{wang-nonbinary}, the non-binary (non)penalized decoder \cite{Liu-nonbinary-journal} and the logarithm-domain fast-Fourier-transform-based Q-ary sum-product algorithm (Log-FFT-QSPA) \cite{burst-error} respectively. All the data points in Figure \ref{iter4} and Figure \ref{time4} are averaged based on generating one million LDPC frames. In Figure \ref{iter4}, we can observe that the proposed ADMM decoder takes more averaged number of iteration than the non-binary (non)penalized decoder in \cite{Liu-nonbinary-journal} in a high SNR region. This could be that more auxiliary variables are required in the proposed decoding model compared to the model in \cite{Liu-nonbinary-journal}. Moreover, the proposed ADMM decoder takes less averaged number of iteration and averaged decoding time than the proximal-ADMM decoder in \cite{wang-nonbinary}. This is mainly because that the proposed decoding model solves much less constraints and variables than the proximal-ADMM decoder \cite{wang-nonbinary}.
In Table \ref{constraints-variable-table}, we compare the number of constraints and variables (the number of rows and columns of $\mathbf{A}$ in the decoding model) in the proposed ADMM decoding model and the decoding model in \cite{wang-nonbinary}.
Obviously, our proposed decoding model solves less constraints and variables, hence our ADMM decoding algorithm is more efficient.

\section{Conclusion}\label{conclusion}
In this paper, an efficient decoder is proposed for non-binary LDPC codes in $\mathbb{F}_{2^q}$ through ADMM method based on the proposed bit embedding technique.
We show that the decoding complexity grows linearly as the non-binary LDPC codes length and the considered length $q$, which outperforms state-of-the-art ADMM decoders.
 Moreover the performance of the proposed decoding algorithm satisfies the codeword-independent property.
Numerical results reveal the outstanding performance in error-correction and complexity compared with other ADMM decoders.

\appendices

\section{Proof of lemma \ref{equivalent_codeword}}\label{proof of equivalent_codeword}
\begin{proof}
Replace $c_i$, $i=1,\cdots,n$ by its bits representation, then we have
\begin{equation}\label{Append1-1}
\begin{aligned}
h_{j,i}c_i&=h_{j,i}\begin{bmatrix}2^{q-1} &\cdots& 2& 1\end{bmatrix} \begin{bmatrix} x_{i,1}\\ \vdots \\ x_{i,q-1}\\ x_{i,q} \end{bmatrix}\\
&=\begin{bmatrix}2^{q-1}h_{j,i} &\cdots& 2h_{j,i}& h_{j,i} \end{bmatrix} \begin{bmatrix} x_{i,1}\\ \vdots \\ x_{i,q-1}\\ x_{i,q} \end{bmatrix}\\
&=\hat{\mathbf{h}}^{(j,i)}\mathbf{x}_{i},
\end{aligned}
\end{equation}
where $\hat{\mathbf{h}}^{(j,i)}=\begin{bmatrix}\hat{h}^{(j,i)}_1 &\cdots & \hat{h}^{(j,i)}_{q-1}&\hat{h}^{(j,i)}_{q}\end{bmatrix}=\begin{bmatrix}2^{q-1}h_{j,i} &\cdots& 2h_{j,i}& h_{j,i} \end{bmatrix}$.

Use the bits representation to denote all the entries in $\hat{\mathbf{h}}^{(j,i)}$, then we get
\begin{equation}\label{Append1-2}
\begin{aligned}
\hat{\mathbf{h}}^{(j,i)}&=\begin{bmatrix}2^{q-1} &\cdots& 2& 1\end{bmatrix}\begin{bmatrix} \hat{h}^{(j,i)}_{1,1} &\cdots& \hat{h}^{(j,i)}_{1,q-1}& \hat{h}^{(j,i)}_{1,q}\\ \vdots&\vdots&\vdots&\vdots \\ \hat{h}^{(j,i)}_{q-1,1} &\cdots& \hat{h}^{(j,i)}_{q-1,q-1}& \hat{h}^{(j,i)}_{q-1,q}\\ \hat{h}^{(j,i)}_{q,1} &\cdots& \hat{h}^{(j,i)}_{q,q-1}& \hat{h}^{(j,i)}_{q,q} \end{bmatrix} \\
&=\begin{bmatrix}2^{q-1} &\cdots& 2& 1\end{bmatrix} \hat{\mathbf{{H}}}_{j,i}.
\end{aligned}
\end{equation}
Combine \eqref{Append1-1} with \eqref{Append1-2}, we can obtain

\begin{equation}
h_{j,i}c_i=\begin{bmatrix}2^{q-1} &\cdots& 2& 1\end{bmatrix}\hat{\mathbf{{H}}}_{j,i}\mathbf{x}_i.
\end{equation}
Follow the bit embedding rule in \eqref{embed_rule}, we can conclude that $\hat{\mathbf{{H}}}_{j,i}\mathbf{x}_i$ is an equivalent binary codeword to $h_{j,i}c_i$.

\end{proof}

\section{Proof of theorem \ref{theorem-codeword-independent}}\label{proof of codeword-independent}
1) \emph{Sketch of the Proof}: We need to prove $\rm{Pr}[\rm{error}|\mathbf{0}]=Pr[\rm{error}|\mathbf{v}]$, where $\mathbf{v}$ is any non-zero codeword. Let

$\mathcal{B}$($\mathbf{v}$):=$\{\mathbf{r}|$ Decoder fails to recover $ \mathbf{v}$ if $\mathbf{r}$ is received$\}$.\\
Then $\rm{Pr}[\rm{error}|\mathbf{v}]=\sum_{\mathbf{r}\in\mathcal{B}(\mathbf{v})}\rm{Pr}[\mathbf{r}|\mathbf{v}]$.

We use the symmetry condition for fields in \cite{Liu-nonbinary-journal} (Definition 48, \cite{Liu-nonbinary-journal}). Let $\mathbf{r}$ and $\mathbf{r}^0$ are received vectors when codeword $\mathbf{v}$ and the all-zeros codeword $\mathbf{0}$ are transmitted over the symmetrical channel. This definition indicates a one-to-one mapping from the received vector $\mathbf{r}$ to a vector $\mathbf{r}^0$ such that the following two statements hold:
\begin{enumerate}[(a)]
\item Pr$[\mathbf{r}|\mathbf{v}]$=Pr$[\mathbf{r}^0|\mathbf{0}]$,
\item $\mathbf{r}\in \mathcal{B}(\mathbf{v})$ if and only if $\mathbf{r}^0\in\mathcal{B}(\mathbf{0})$.
\end{enumerate}
Statement (a) is directly proved according to the definition of the symmetry condition. In the following, we show that statement (b) also holds once we proved that
\begin{equation*}
\begin{aligned}
\rm{Pr}[error|\mathbf{v}]&=\sum_{\mathbf{r}\in\mathcal{B}(\mathbf{v})}\rm{Pr}[\mathbf{r}|\mathbf{v}]\\
&=\sum_{\mathbf{r}^0\in\mathcal{B}(\mathbf{0})}\rm{Pr}[\mathbf{r}^0|\mathbf{0}]=\rm{Pr}[error|\mathbf{0}].
\end{aligned}
\end{equation*}

2) \emph{Proof of statement (b):} We define the concept of relative vector for the bit embedding.
\begin{definition}\label{relative vector}
Define $\mathrm{R}_{\mathbf{v}}(\pmb{\beta})$ is the relative vector of $\pmb\beta$ with respect to the binary vector $\mathbf{v}$ if it satisfies
\begin{equation}
\left(\mathrm{R}_{\mathbf{v}}(\pmb{\beta})\right)_{i}=\left\{\begin{array}{ll}
\beta_{i}, & v_{i}=0, \\
1-\beta_{i}, & v_{i}=1.
\end{array}\right.
\end{equation}
Vector $\mathbf{v}=[\mathbf{x};\mathbf{u}]$, where $\mathbf{x}$ is the bit embedding result of the transmitted non-binary LDPC codeword $\mathbf{c}$ and $\mathbf{u}$ is the corresponding binary auxiliary variables.
\end{definition}

 In \eqref{LP-relax_aux}, an auxiliary variable $\mathbf{e}_1$ is introduced to change the inequality constraint \eqref{LP-relax-b} to the equality constraint \eqref{LP-relax-aux-b}. Moreover, observing \eqref{four_inequalities}, four auxiliary variables are required for one three-variables inequalities set. Hence we define the following mapping operators $\mathcal{M}_{\mathbf{v}_\tau}(\cdot)$ and $\mathcal{T}_{\mathbf{v}}(\cdot)$ for the introduced auxiliary variables. These two operators can also be applied to the Lagrangian multiplier $\mathbf{y}_1$. Moreover, we also define a mapping operator $\mathcal{K}_{\mathbf{v}}(\mathbf{y}_2)$ for $\mathbf{y}_2$.
 \begin{definition}\label{def-M-T-operator}
 Let $\mathbf{e}_{1,\tau}=[e_{1,{\tau}_1},e_{1,{\tau}_2},e_{1,{\tau}_3},e_{1,{\tau}_4}]$ be the auxiliary variables which change the inequality constraints corresponding to the $\tau$th three-variables parity-check equation to the equality constraints. Denote $\mathbf{v}_\tau=\mathbf{Q}_\tau\mathbf{v}, \tau=1,\cdots,\Gamma_c$, as the corresponding binary variables. Then we define the following mapping operator for $\mathbf{e}_{1,\tau}$
 \begin{equation}
 \mathcal{M}_{\mathbf{v}_{\tau}}\left(\mathbf{e}_{1,\tau}\right)\!=\!\left\{\begin{array}{llll}
{\left[\begin{array}{llll}
e_{1,\tau_{2}}\!\!\!\!\! & e_{1,\tau_{1}}\!\!\!\!\! & e_{1,\tau_{4}}\!\!\!\!\! & e_{1,\tau_{3}}
\end{array}\right]^{T},} &\!\!\!\! \mathbf{v}_{\tau}\!=\! \left[\!\!\begin{array}{lll}
1\!\!\! &\!\!\! 1 &\!\!\! 0
\end{array}\!\!\right]^{T}, \\
{\left[\begin{array}{llll}
e_{1,\tau_{3}}\!\!\!\!\! & e_{1,\tau_{4}}\!\!\!\!\! & e_{1,\tau_{1}}\! \!\!\!\!& e_{1,\tau_{2}}
\end{array}\right]^{T},} &\!\!\!\! \mathbf{v}_{\tau}\!=\!\left[\!\!\begin{array}{lll}
1\!\!\! &\!\!\! 0 & \!\!\! 1
\end{array}\!\!\right]^{T} ,\\
{\left[\begin{array}{llll}
e_{1,\tau_{4}} \!\!\!\!\! & e_{1,\tau_{3}} \!\!\!\!\! & e_{1,\tau_{2}}\!\!\!\!\! & e_{1,\tau_{1}}
\end{array}\right]^{T},} &\!\!\!\! \mathbf{v}_{\tau}\!=\!\left[\!\!\begin{array}{lll}
0 \!\!\!&\!\!\! 1 &\!\!\! 1
\end{array}\!\!\right]^{T}, \\
\mathbf{e}_{1,\tau}, &\!\!\!\! \mathbf{v}_{\tau}\!=\!\left[\!\!\begin{array}{lll}
0\!\!\! &\!\!\! 0 &\!\!\! 0
\end{array}\!\!\right]^{T}.
\end{array}\right.
 \end{equation}
 Furthermore, we define
 \[\mathcal{T}_{\mathbf{v}}(\mathbf{e}_1)\!=\!\left[\mathcal{M}_{\mathbf{v}_{1}}\left(\mathbf{e}_{1,1}\right) ; \cdots ; \mathcal{M}_{\mathbf{v}_{\tau}}\left(\mathbf{e}_{1,\tau}\right) ; \cdots ; \mathcal{M}_{\mathbf{v}_{\Gamma_{c}}}\left(\mathbf{e}_{1,\Gamma_{c}}\right)\right]^{T}.\]

 Moreover, we define the following mapping operator for $\mathbf{y}_2$
 \begin{equation}
\left(\mathcal{K}_{\mathbf{v}}(\mathbf{y}_2)\right)_{i}=\left\{\begin{array}{ll}
y_{2,i}, & v_{i}=0, \\
-y_{2,i}, & v_{i}=1.
\end{array}\right.
\end{equation}
 \end{definition}

 \begin{lemma}\label{lemma-symmetry-v0k-v0k+1}
 In Algorithm \ref{ADMM-QP}, let $\{\mathbf{v}^k,\mathbf{e}_1^k,\mathbf{e}_2^k,\mathbf{y}_1^k,\mathbf{y}_2^k\}$ be the vectors in the $k$th iteration when decoding $\mathbf{r}$. Let $\{\mathbf{v}^{0,k},\mathbf{e}_1^{0,k},\mathbf{e}_2^{0,k},\mathbf{y}_1^{0,k},\mathbf{y}_2^{0,k}\}$ be the vectors after the $k$th iteration when decoding $\mathbf{r}^0$. If $\mathbf{v}^{0,k}=\rm{R}_{\mathbf{v}}(\mathbf{v}^{k})$, $\mathbf{e}_1^{0,k}=\mathcal{T}_{\mathbf{v}}(\mathbf{e}_1^{k})$, $\mathbf{e}_2^{0,k}=\rm{R}_{\mathbf{v}}(\mathbf{e}_2^{k})$, $\mathbf{y}_1^{0,k}=\mathcal{T}_{\mathbf{v}}(\mathbf{y}_1^{k})$, $\mathbf{y}_2^{0,k}=\mathcal{K}_{\mathbf{v}}(\mathbf{y}_2^{k})$, then $\mathbf{v}^{0,k+1}=\rm{R}_{\mathbf{v}}(\mathbf{v}^{k+1})$, $\mathbf{e}_1^{0,k+1}=\mathcal{T}_{\mathbf{v}}(\mathbf{e}_1^{k+1})$, $\mathbf{e}_2^{0,k+1}=\rm{R}_{\mathbf{v}}(\mathbf{e}_2^{k+1})$, $\mathbf{y}_1^{0,k+1}=\mathcal{T}_{\mathbf{v}}(\mathbf{y}_1^{k+1})$, $\mathbf{y}_2^{0,k+1}=\mathcal{K}_{\mathbf{v}}(\mathbf{y}_2^{k+1})$.
 \end{lemma}
 \begin{proof}
 See Appendix \ref{proof of symmetry-v0k-v0k+1}.
 \end{proof}

 \begin{lemma}\label{lemma-v0=Rv}
 Let $\hat{\mathbf{v}}$ and $\hat{\mathbf{v}}^0$ denote the output of the Algorithm \ref{ADMM-QP} when $\mathbf{r}$ and $\mathbf{r}^0$ are received respectively. Then there exists $\hat{\mathbf{v}}^0=\rm{R}_{\mathbf{v}}(\hat{\mathbf{v}})$.
 \end{lemma}
 \begin{proof}
 Let $\{\mathbf{e}_1^{0,0},\mathbf{e}_2^{0,0},\mathbf{y}_1^{0,0},\mathbf{y}_2^{0,0}\}$ and $\{\mathbf{e}_1^{0},\mathbf{e}_2^{0},\mathbf{y}_1^{0},\mathbf{y}_2^{0}\}$ be the initial values when decode the all-zeros codeword and general codeword $\mathbf{c}$ using Algorithm \ref{ADMM-QP}. We setting them to satisfy $\mathbf{e}_1^{0,0}=\mathcal{T}_{\mathbf{v}}(\mathbf{e}_1^{0})$, $\mathbf{e}_2^{0,0}=\rm{R}_{\mathbf{v}}(\mathbf{e}_2^{0})$, $\mathbf{y}_1^{0,0}=\mathcal{T}_{\mathbf{v}}(\mathbf{y}_1^{0})$, $\mathbf{y}_2^{0,0}=\mathcal{K}_{\mathbf{v}}(\mathbf{y}_2^{0})$. Then in the first iteration, by Lemma \ref{lemma-symmetry-v0k-v0k+1}, we have $\mathbf{v}^{0,1}=\rm{R}_{\mathbf{v}}(\mathbf{v}^1)$, $\mathbf{{e}}_1^{0,1}=\mathcal{T}_{\mathbf{v}}(\mathbf{e}_1^{1})$, $\mathbf{e}_2^{0,1}=\rm{R}_{\mathbf{v}}(\mathbf{e}_2^{1})$, $\mathbf{y}_1^{0,1}=\mathcal{T}_{\mathbf{v}}(\mathbf{y}_1^{1})$, $\mathbf{y}_2^{0,1}=\mathcal{K}_{\mathbf{v}}(\mathbf{y}_2^{1})$. Continuing with the iteration, we can get $\mathbf{v}^{0,k}=\rm{R}_{\mathbf{v}}(\mathbf{v}^k)$, $\mathbf{{e}}_1^{0,k}=\mathcal{T}_{\mathbf{v}}(\mathbf{e}_1^{k})$, $\mathbf{e}_2^{0,k}=\rm{R}_{\mathbf{v}}(\mathbf{e}_2^{k})$, $\mathbf{y}_1^{0,k}=\mathcal{T}_{\mathbf{v}}(\mathbf{y}_1^{k})$, $\mathbf{y}_2^{0,k}=\mathcal{K}_{\mathbf{v}}(\mathbf{y}_2^{k})$.

 For the termination condition in Algorithm \ref{ADMM-QP}, we have the following derivations
 \begin{equation}
 \begin{aligned}
 \|\mathbf{Av}^k+\mathbf{e}_1^k-\mathbf{b}\|^2_2&=\sum^{\Gamma_c}_{\tau=1}\|\mathbf{Tv}^{k}_{\tau}+\mathbf{e}_{1,\tau}^{k}-\mathbf{w}\|^2_2\\
 &\overset{a}{=}\sum^{\Gamma_c}_{\tau=1}\|\mathbf{Tv}^{0,k}_{\tau}+\mathbf{e}_{1,\tau}^{0,k}-\mathbf{w}\|^2_2\\
 &=\|\mathbf{Av}^{0,k}+\mathbf{e}_1^{0,k}-\mathbf{b}\|^2_2,
 \end{aligned}
 \end{equation}
 \begin{equation}
 \|\mathbf{v}^k-\mathbf{e}_2^k\|^2_2\overset{b}{=}\|\mathbf{v}^{0,k}-\mathbf{e}_2^{0,k}\|^2_2
 \end{equation}
where the equalities "a" and "b" hold since $\mathbf{v}^{0,k}=\rm{R}_{\mathbf{v}}(\mathbf{v}^k)$, $\mathbf{{e}}_1^{0,k}=\mathcal{T}_{\mathbf{v}}(\mathbf{e}_1^{k})$, $\mathbf{e}_2^{0,k}=\mathcal{K}_{\mathbf{v}}(\mathbf{e}_2^{k})$, and we can find that the vector $\mathbf{Tv}^{k}_{\tau}+\mathbf{e}_{1,\tau}^{k}-\mathbf{w}$ can be obtained from the permutation of the vector $\mathbf{Tv}^{0,k}_{\tau}+\mathbf{e}_{1,\tau}^{0,k}-\mathbf{w}$. This indicates that both of the decoding procedures in Algorithm \ref{ADMM-QP} for $\mathbf{r}$ and $\mathbf{r}^0$ should always be terminated at the same time. Hence we conclude $\mathbf{v}^0=\rm{R}_{\mathbf{v}}(\mathbf{v})$.
 \end{proof}

 Due to Lemma \ref{lemma-v0=Rv}, if the decoder recovers a codeword $\mathbf{c}$ from $\mathbf{r}$, it indicates that the decoded vector $\hat{\mathbf{v}}$ is the embedding of $\mathbf{c}$ based on Definition \ref{def-bit-embed}. Therefore by Definition \ref{relative vector}, $\rm{R}_{\mathbf{v}}(\hat{\mathbf{v}})$ is the embedding of $\mathbf{0}$ from $\mathbf{r}^0$. On the flip side, if the decoder fails to recover codeword $\mathbf{c}$ from $\mathbf{r}$, then $\rm{R}_{\mathbf{v}}(\hat{\mathbf{v}})$ is not a integral vector. This indicates that the decoder cannot recover $\mathbf{0}$ from $\mathbf{r}^0$. Combine the two augments, we can conclude that the decoder can recover $\mathbf{c}$ from $\mathbf{r}$ if and only if it can recover $\mathbf{0}$ from $\mathbf{r}^0$, which completes the proof.

\section{Proof of lemma \ref{lemma-symmetry-v0k-v0k+1}}\label{proof of symmetry-v0k-v0k+1}
To prove Lemma \ref{lemma-symmetry-v0k-v0k+1}, we provide the following corollary based on Definition \ref{def-M-T-operator}.
\begin{corollary}\label{tlMv(e1)}
The $\tau$th three-variables parity-check equation $\mathbf{e}_{1,\tau}$ and its mapping operator $\mathcal{M}_{\mathbf{v}_\tau}(\mathbf{e}_{1,\tau})$ have the following property
\begin{equation}\label{tlMv=tle1}
\mathbf{t}_{\ell}^{T} \mathcal{M}_{\mathbf{v}_{\tau}}\left(\mathbf{e}_{1,\tau}\right)=\left\{\begin{array}{cl}
\mathbf{t}_{\ell}^{T} \mathbf{e}_{1,\tau}, & v_{\tau_{\ell}}=0 \\
-\mathbf{t}_{\ell}^{T} \mathbf{e}_{1,\tau}, & v_{\tau_{\ell}}=1
\end{array}\right.{}
\end{equation}
where $\mathbf{t}_\ell,\ell=1,2,3$, denotes the $\ell$th column of matrix $\mathbf{T}$.
\end{corollary}
\begin{proof}
Based on the Definition \ref{def-M-T-operator}, there are four cases for the relationship between $\mathbf{e}_{1,\tau}$ and the mapping operator $\mathcal{M}_{\mathbf{v}_{\tau}}(\mathbf{e}_{1,\tau})$. When $\mathbf{v}_{\tau}=[1,1,0]^T$, $\mathcal{M}_{\mathbf{v}_{\tau}}(\mathbf{e}_{1,\tau})=[e_{1,\tau_{2}},e_{1,\tau_{1}},e_{1,\tau_{4}},e_{1,\tau_{3}}]^T$. Hence, when $\mathbf{v}_{\tau_1}=1$, $\mathbf{v}_{\tau_2}=1$ and $\mathbf{v}_{\tau_3}=0$ respectively, we can obtain the following equation
\[\mathbf{t}^T_1{M}_{\mathbf{v}_{\tau}}(\mathbf{e}_{1,\tau})=-e_{1,\tau_{1}}+e_{1,\tau_{2}}+e_{1,\tau_{3}}-e_{1\tau_{4}}=-\mathbf{t}^T_1\mathbf{e}_{1,\tau},\]
\[\mathbf{t}^T_2{M}_{\mathbf{v}_{\tau}}(\mathbf{e}_{1,\tau})=e_{1,\tau_{1}}-e_{1,\tau_{2}}+e_{1,\tau_{3}}-e_{1\tau_{4}}=-\mathbf{t}^T_2\mathbf{e}_{1,\tau},\]
\[\mathbf{t}^T_3{M}_{\mathbf{v}_{\tau}}(\mathbf{e}_{1,\tau})=-e_{1,\tau_{1}}-e_{1,\tau_{2}}+e_{1,\tau_{3}}+e_{1\tau_{4}}=-\mathbf{t}^T_3\mathbf{e}_{1,\tau},\]
which can be further written as
\begin{equation}
\mathbf{t}_{\ell}^{T} \mathcal{M}_{\mathbf{v}_{\tau}}\left(\mathbf{e}_{1,\tau}\right)=\left\{\begin{aligned}
\mathbf{t}_{\ell}^{T} \mathbf{e}_{1,\tau}, &\ \ v_{\tau_{\ell}}=0, \\
-\mathbf{t}_{\ell}^{T} \mathbf{e}_{1,\tau}, &\ \ v_{\tau_{\ell}}=1,
\end{aligned}\right.
\end{equation}
where $\ell=1,2,3.$

Therefore, through similar deduction, it can be proved that \eqref{tlMv=tle1} also holds when $\mathbf{v}_{\tau}=[1,0,1]^T, [0,1,1]^T$ and $[0,0,0]^T$. This completes the proof.
\end{proof}

Now we are ready to prove Lemma $\ref{lemma-symmetry-v0k-v0k+1}$.
\begin{proof}
According to \eqref{v-update-parallel}, we have
\begin{equation}
\begin{aligned}
 v_i^{k+1}=\frac{1}{s_i-\frac{\alpha}{\mu}+1}&\left(\mathbf{a}_i^T(\mathbf{b}-\mathbf{e}_1^k-\frac{\mathbf{y}_1^k}{\mu})\right.\\
 &\left.+e_{2,i}^k-\frac{y_{2,i}^k}{\mu}-\frac{2q_i+\alpha}{2\mu}\right).
 \end{aligned}
\end{equation}
We first consider the case that $v_i=1$ is transmitted. Then we have
\begin{equation}\label{1-vik+1}
\begin{aligned}
 1-v_i^{k+1}\overset{a}{=}&\frac{1}{s_i-\frac{\alpha}{\mu}+1}\left(\mathbf{a}_i^T(\mathbf{b}+\mathbf{e}_1^k+\frac{\mathbf{y}_1^k}{\mu})\right.\\
 &\left.+1-e_{2,i}^k+\frac{y_{2,i}^k}{\mu}+\frac{2q_i-\alpha}{2\mu}\right)\\
 \overset{b}{=}&\frac{1}{s_i-\frac{\alpha}{\mu}+1}\left(\mathbf{a}_i^T(\mathbf{b}-\mathbf{e}_1^{0,k}-\frac{\mathbf{y}_1^{0,k}}{\mu})\right.\\
 &\left.+e_{2,i}^{0,k}-\frac{y_{2,i}^{0,k}}{\mu}-\frac{2q_i^0+\alpha}{2\mu}\right),\\
 \end{aligned}
\end{equation}
where $q_i^0\in\mathbf{q}^0=[\pmb\gamma^0;\mathbf{0}]$ and $\pmb\gamma^0$ is the LLR vector when the all-zeros codeword is transmitted.

In \eqref{1-vik+1}, the equality "a" holds since $\mathbf{a}_i^T\mathbf{b}=2d_i$ and $s_i=4d_i$, where $d_i$ denotes the nonzero numbers in the $i$th column of the parity-check matrix $[\hat{\mathbf{H}}_1;\cdots;\hat{\mathbf{H}}_m]$. The equality "b" holds if $q_i^0=-q_i$, $\mathbf{a}_i^T\mathbf{e}_1^k=-\mathbf{a}_i^T\mathbf{e}_1^{0,k}$, $\mathbf{a}_i^T\mathbf{y}_1^k=-\mathbf{a}_i^T\mathbf{y}_1^{0,k}$. It is obvious that $q_i^0=-q_i$ when $v_i=1$ since the channel is assumed to be symmetrical. And $e_{2,i}^{0,k}=\rm{R}_{\mathbf{v}}(e_{2,i}^{k})=1-e_{2,i}^{k}$, $y_{2,i}^{0,k}=\mathcal{K}_{\mathbf{v}}(y_{2,i}^k)=-y_{2,i}^k$ holds when $v_i=1$. Now we need to verify $\mathbf{a}_i^T\mathbf{e}_1^k=-\mathbf{a}_i^T\mathbf{e}_1^{0,k}$ and $\mathbf{a}_i^T\mathbf{y}_1^k=-\mathbf{a}_i^T\mathbf{y}_1^{0,k}$. Observing \eqref{A}, we can find that there are $\lfloor\mathbf{t}_\ell \rfloor$ sub-vectors same as $\mathbf{t}_\ell$ in $\mathbf{a}_i$, where $\lfloor \mathbf{t}_\ell\rfloor$ is the number of vector same as $\mathbf{t}_\ell$. Moreover, we denote $\mathbf{e}_{1,\kappa_\ell}$ as the length-4 auxiliary variables corresponding to the $\kappa_\ell$th sub-vectors which is the same as $\mathbf{t}_\ell$, $\kappa_\ell=1,\cdots,\lfloor \mathbf{t}_\ell\rfloor$. Based on the above augment, when $v_i=1$, we have
\begin{equation}\label{ae1k-ae10k}
\begin{aligned}
\mathbf{a}_{i}^{T} \mathbf{e}_1^{k} &=\sum_{\kappa_{1}=1}^{\left\lfloor\mathbf{t}_{1}\right\rfloor} \mathbf{t}_{1}^{T} \mathbf{e}_{1,\kappa_{1}}^{k}+\sum_{\kappa_{2}=1}^{\left\lfloor\mathbf{t}_{2}\right\rfloor} \mathbf{t}_{2}^{T} \mathbf{e}_{1,\kappa_{2}}^{k}+\sum_{\kappa_{3}=1}^{\left\lfloor\mathbf{t}_{3}\right\rfloor} \mathbf{t}_{3}^{T} \mathbf{e}_{1,\kappa_{3}}^{k} \\
&=-\sum_{\kappa_{1}=1}^{\left\lfloor\mathbf{t}_{1}\right\rfloor} \mathbf{t}_{1}^{T} \mathbf{e}_{1,\kappa_{1}}^{0,k}-\sum_{\kappa_{2}=1}^{\left\lfloor\mathbf{t}_{2}\right\rfloor} \mathbf{t}_{2}^{T} \mathbf{e}_{1,\kappa_{2}}^{0,k}-\sum_{\kappa_{3}=1}^{\left\lfloor\mathbf{t}_{3}\right\rfloor} \mathbf{t}_{3}^{T} \mathbf{e}_{1,\kappa_{3}}^{0,k} \\
&=-\mathbf{a}_{i}^{T} \mathbf{e}_1^{0,k}.
\end{aligned}
\end{equation}
In \eqref{ae1k-ae10k}, the second equality hods since $\mathbf{t}_\ell^T\mathbf{e}_{1,\kappa_\ell}^k=-\mathbf{t}_\ell^T\mathbf{e}_{1,\kappa_\ell}^{0,k},\ell=1,2,3$, when $v_i=1$ based on Corollary \ref{tlMv(e1)}. With similar derivation of \ref{tlMv(e1)}, we can also obtain that $\mathbf{a}_i^T\mathbf{y}_1^k=-\mathbf{a}_i^T\mathbf{y}_1^{0,k}$ when $v_i=1$. Hence, equation \eqref{1-vik+1} holds, which means that when $v_i=1$, we have
\begin{equation}\label{v0k+1=1-vk+1}
v_i^{0,k+1}=1-v_i^{k+1}
\end{equation}

For the other case when $v_i=0$ is transmitted over the channel, we have
\begin{equation}\label{vk+1=v0k+1}
\begin{aligned}
 v_i^{k+1}=\frac{1}{s_i-\frac{\alpha}{\mu}+1}&\left(\mathbf{a}_i^T(\mathbf{b}-\mathbf{e}_1^{0,k}-\frac{\mathbf{y}_1^{0,k}}{\mu})\right.\\
 &\left.+e_{2,i}^{0,k}-\frac{y_{2,i}^{0,k}}{\mu}-\frac{2q_i^0+\alpha}{2\mu}\right),
 \end{aligned}
\end{equation}
where the equality holds since $q^0_i=q_i$ when $v_i=0$ under the symmetrical channel. And $e_{2,i}^{0,k}=\rm{R}_{\mathbf{v}}(e_{2,i}^{k})=e_{2,i}^{k}$, $y_{2,i}^{0,k}=\mathcal{K}_{\mathbf{v}}(y_{2,i}^k)=y_{2,i}^k$ holds when $v_i=0$. Furthermore, base on the Corollary \ref{tlMv(e1)}, when $v_i=0$, we have
\begin{equation}\label{ae1k+ae10k}
\begin{aligned}
\mathbf{a}_{i}^{T} \mathbf{e}_1^{k} &=\sum_{\kappa_{1}=1}^{\left\lfloor\mathbf{t}_{1}\right\rfloor} \mathbf{t}_{1}^{T} \mathbf{e}_{1,\kappa_{1}}^{k}+\sum_{\kappa_{2}=1}^{\left\lfloor\mathbf{t}_{2}\right\rfloor} \mathbf{t}_{2}^{T} \mathbf{e}_{1,\kappa_{2}}^{k}+\sum_{\kappa_{3}=1}^{\left\lfloor\mathbf{t}_{3}\right\rfloor} \mathbf{t}_{3}^{T} \mathbf{e}_{1,\kappa_{3}}^{k} \\
&=\sum_{\kappa_{1}=1}^{\left\lfloor\mathbf{t}_{1}\right\rfloor} \mathbf{t}_{1}^{T} \mathbf{e}_{1,\kappa_{1}}^{0,k}+\sum_{\kappa_{2}=1}^{\left\lfloor\mathbf{t}_{2}\right\rfloor} \mathbf{t}_{2}^{T} \mathbf{e}_{1,\kappa_{2}}^{0,k}+\sum_{\kappa_{3}=1}^{\left\lfloor\mathbf{t}_{3}\right\rfloor} \mathbf{t}_{3}^{T} \mathbf{e}_{1,\kappa_{3}}^{0,k} \\
&=\mathbf{a}_{i}^{T} \mathbf{e}_1^{0,k}.
\end{aligned}
\end{equation}
Through a similar derivation of \eqref{ae1k+ae10k}, we can also get $\mathbf{a}_i^T\mathbf{y}_1^k=\mathbf{a}_i^T\mathbf{y}_1^{0,k}$ when $v_i=0$. Hence, equation \eqref{vk+1=v0k+1} holds, which means that when $v_i=0$,
\begin{equation}\label{v0k+1=vk+1}
v_i^{0,k+1}=v_i^{k+1}
\end{equation}
Combine \eqref{v0k+1=1-vk+1} and \eqref{v0k+1=vk+1}, we get
\begin{equation}\label{v0k+1=Rvk+1}
\mathbf{v}^{0,k+1}=\rm{R}_{\mathbf{v}}(\mathbf{v}^{k+1})
\end{equation}

Next, we verify that $\mathbf{e}_1^{0,k+1}=\mathcal{T}_{\mathbf{v}}(\mathbf{e}_1^{k+1})$. Firstly, we consider the relationship between $\mathbf{e}_{1,\tau}^{k+1}$ and $\mathbf{e}_{1,\tau}^{0,k+1}$. Since $\mathbf{y}_{1,\tau}^{0,k}=\mathcal{M}_{\mathbf{v}_\tau}(\mathbf{y}_{1,\tau}^{k})$, there exists $\mathbf{y}_{1,\tau}^{0,k}=[{y}_{1,\tau_2}^k,{y}_{1,\tau_1}^k,{y}_{1,\tau_4}^k,{y}_{1,\tau_3}^k]^T$ when $\mathbf{v}_\tau=[1,1,0]^T$. Hence, when $\mathbf{v}_\tau=[1,1,0]^T$, we have
\begin{equation}
\begin{aligned}
\mathbf{e}_{1,\tau}^{0,{k+1}} &=\prod_{[0,+\infty)}\left(\mathbf{w}-\mathbf{T} \mathbf{v}_{\tau}^{0,{k+1}}-\frac{\mathbf{y}_{1,\tau}^{0,k} }{\mu}\right) \\
&=\prod_{[0,+\infty)}\left(\mathbf{w}-\mathbf{T}\left[\begin{array}{c}
1-v_{\tau_{1}}^{k+1} \\
1-v_{\tau_{2}}^{k+1} \\
v_{\tau_{3}}^{k+1}
\end{array}\right]-\frac{1}{\mu}\left[\begin{array}{c}
y_{1,\tau_{2}}^{k} \\
y_{1,\tau_{1}}^{k} \\
y_{1,\tau_{4}}^{k} \\
y_{1,\tau_{3}}^{k}
\end{array}\right]\right) \\
&=\prod_{[0,+\infty)}\left[\begin{array}{c}
\quad v_{\tau_{1}}^{k+1}-v_{\tau_{2}}^{k+1}+v_{\tau_{3}}^{k+1}-\frac{y_{1,\tau_{2}}^{k}}{\mu} \\
-v_{\tau_{1}}^{k+1}+v_{\tau_{2}}^{k+1}+v_{\tau_{3}}^{k+1}-\frac{y_{1,\tau_{1}}^{k}}{\mu} \\
2-\left(v_{\tau_{1}}^{k+1}+v_{\tau_{2}}^{k+1}+v_{\tau_{3}}^{k+1}\right)-\frac{y_{1,\tau_{4}}^{k}}{\mu} \\
v_{\tau_{1}}^{k+1}+v_{\tau_{2}}^{k+1}-v_{\tau_{3}}^{k+1}-\frac{y_{1,\tau 3}^{k}}{\mu}
\end{array}\right]
\end{aligned}
\end{equation}
Then, $\mathbf{e}_{1,\tau}^{0,k+1}$ can be further derived as
\begin{equation}
\begin{aligned}
\mathbf{e}_{1,\tau}^{0,{k+1}} &=\prod_{[0,+\infty)}\left[\begin{array}{c}
w_{2}-\left(-v_{\tau_{1}}^{k+1}+v_{\tau_{2}}^{k+1}-v_{\tau_{3}}^{k+1}\right)-\frac{y_{1,\tau_{2}}^{k}}{\mu} \\
w_{1}-\left(v_{\tau_{1}}^{k+1}-v_{\tau_{2}}^{k+1}-v_{\tau_{3}}^{k+1}\right)-\frac{y_{1,\tau_{1}}^{k}}{\mu} \\
w_{4}-\left(v_{\tau_{1}}^{k+1}+v_{\tau_{2}}^{k+1}+v_{\tau_{3}}^{k+1}\right)-\frac{y_{1,\tau_{4}}^{k}}{\mu} \\
w_{3}-\left(-v_{\tau_{1}}^{k+1}-v_{\tau_{2}}^{k+1}+v_{\tau_{3}}^{k+1}\right)-\frac{y_{1,\tau_{3}}^{k}}{\mu}
\end{array}\right] \\
&=\left[e_{1,\tau_{2}}^{k+1}, e_{1,\tau_{1}}^{k+1}, e_{1,\tau_{4}}^{k+1}, e_{1,\tau_{3}}^{k+1}\right],
\end{aligned}
\end{equation}
i.e.,
\begin{equation}\label{e1tau0k+1=Me1tauk+1}
\mathbf{e}_{1,\tau}^{0,k+1}=\mathcal{M}_{\mathbf{v}_{\tau}}(\mathbf{e}_{1,\tau}^{k+1}).
\end{equation}
With a similar derivation, we can see \eqref{e1tau0k+1=Me1tauk+1} also holds when $\mathbf{v}_\tau=[1,0,1]^T, [0,1,1]^T, [0,0,0]^T$. Since $\mathbf{e}_1^{k+1}$ is cascaded by $\mathbf{e}_{1,\tau}^{k+1}, \tau=1,\cdots,\Gamma_c$, and $\mathbf{e}_1^{0,k+1}$ is also cascaded by $\mathbf{e}_{1,\tau}^{0,k+1}$ in the same way, we have
\begin{equation}\label{e10k+1=Me1k+1}
\mathbf{e}_{1}^{0,k+1}=\mathcal{T}_{\mathbf{v}}(\mathbf{e}_{1}^{k+1}).
\end{equation}

Next, we verify that $\mathbf{e}_2^{0,k+1}=\rm{R}_{\mathbf{v}}(\mathbf{e}_2^{k+1})$. According to \eqref{e2_update_parallel}, we denote $e^{k+1}_{2,i}=\prod_{[0,1]}(\bar{e}^{k+1}_{2,i})$, and $\bar{e}^{k+1}_{2,i}=\left(v_i^{k+1}+\frac{y_{2,i}^{k}}{\mu}\right)$. When $v_i=1$, we have
\begin{equation}\label{e2i0k+1=1-e2ik+1-detail}
 \bar{e}^{0,k+1}_{2,i}=v_i^{0,k+1}+\frac{y_{2,i}^{0,k}}{\mu}=1-\left(v_i^{k+1}+\frac{y_{2,i}^{k}}{\mu} \right),
 \end{equation}
 where the second equality holds since $v_i^{0,k+1}={\rm{R}}_{\mathbf{v}}(v_i^{k+1})=1-v_i^{k+1}$ and $y_{2,i}^{0,k}=\mathcal{K}_{\mathbf{v}}(y_{2,i}^{k})=-y_{2,i}^{k}$ when $v_i=1$. Thus from equation \eqref{e2i0k+1=1-e2ik+1-detail} we have $\bar{e}^{0,k+1}_{2,i}=1-\bar{e}^{k+1}_{2,i}$ when $v_i=1$. Moreover, since $\prod_{[0,1]}(\bar{e}^{0,k+1}_{2,i})=1-\prod_{[0,1]}(\bar{e}^{k+1}_{2,i})$, we have
 \begin{equation}\label{e2i0k+1=1-e2ik+1}
 e^{0,k+1}_{2,i}=1-e^{k+1}_{2,i}.
 \end{equation}

 For the other case when $v_i=0$ is transmitted over the channel, there exist
 \begin{equation}
 \bar{e}^{0,k+1}_{2,i}=v_i^{0,k+1}+\frac{y_{2,i}^{0,k}}{\mu}=v_i^{k+1}+\frac{y_{2,i}^{k}}{\mu} ,
 \end{equation}
where the second equality holds since $v_i^{0,k+1}={\rm{R}}_{\mathbf{v}}(v_i^{k+1})=v_i^{k+1}$ and $y_{2,i}^{0,k}=\mathcal{K}_{\mathbf{v}}(y_{2,i}^{k})=y_{2,i}^{k}$ when $v_i=0$. With a similar derivation of \eqref{e2i0k+1=1-e2ik+1}, when $v_i=0$, we have
\begin{equation}\label{e2i0k+1=e2ik+1}
 e^{0,k+1}_{2,i}=e^{k+1}_{2,i}.
 \end{equation}
 Combine \eqref{e2i0k+1=1-e2ik+1} and \eqref{e2i0k+1=e2ik+1}, we get
 \begin{equation}
 \mathbf{e}^{0,k+1}_{2}={\rm{R}}_{\mathbf{v}}(\mathbf{e}^{k+1}_{2}).
 \end{equation}

 Next, we verify that $\mathbf{y}_1^{0,k+1}=\mathcal{T}_{\mathbf{v}}(\mathbf{y}_1^{k+1})$. We first consider the relationship between $\mathbf{y}_{1,\tau}^{k+1}$ and $\mathbf{y}_{1,\tau}^{0,k+1}$. Since $\mathbf{e}_{1,\tau}^{0,k+1}=\mathcal{M}_{\mathbf{v}_{\tau}}(\mathbf{e}_{1,\tau}^{k+1})$ based on \eqref{e1tau0k+1=Me1tauk+1}, there are $\mathbf{e}_{1,\tau}^{0,k+1}=[{e}_{1,\tau_2}^{k+1},{e}_{1,\tau_1}^{k+1},{e}_{1,\tau_4}^{k+1},{e}_{1,\tau_3}^{k+1}]^T$ when $\mathbf{v}_\tau=[1,1,0]^T$. Thus, when $\mathbf{v}_\tau=[1,1,0]^T$, from \eqref{subproblems-c} and \eqref{v0k+1=Rvk+1} we have
\begin{equation}
\begin{aligned}
\mathbf{y}_{1,\tau}^{0,{k+1}} &=\mathbf{y}_{1,\tau}^{0,{k}}+\mu\left(\mathbf{T v}_{\tau}^{0,{k+1}}+\mathbf{e}_{1,\tau}^{0,{k+1}}-\mathbf{w}\right) \\
&=\left[\begin{array}{c}
y_{1,\tau_{2}}^{k} \\
y_{1,\tau_{1}}^{k} \\
y_{1,\tau_{4}}^{k} \\
y_{1,\tau_{3}}^{k}
\end{array}\right]+\mu\left(\mathbf{T}\left[\begin{array}{c}
1-v_{\tau_{1}}^{k+1} \\
1-v_{\tau_{2}}^{k+1} \\
v_{\tau_{3}}^{k+1}
\end{array}\right]+\left[\begin{array}{c}
e_{1,\tau_{2}}^{k+1} \\
e_{1,\tau_{1}}^{k+1} \\
e_{1,\tau_{4}}^{k+1} \\
e_{1,\tau_{3}}^{k+1}
\end{array}\right]-\mathbf{w}\right) \\
&\left.=\left[\begin{array}{l}
y_{1,\tau_{2}}^{k}+\mu\left(-v_{\tau_{1}}^{k+1}+v_{\tau_{2}}^{k+1}-v_{\tau_{3}}^{k+1}+e_{1,\tau_{2}}^{k+1}\right) \\
y_{1,\tau_{1}}^{k}+\mu\left(v_{\tau_{1}}^{k+1}-v_{\tau_{2}}^{k+1}-v_{\tau_{3}}^{k+1}+e_{1,\tau_{1}}^{k+1}\right) \\
y_{1,\tau_{4}}^{k}+\mu\left(v_{\tau_{1}}^{k+1}+v_{\tau_{2}}^{k+1}+v_{\tau_{3}}^{k+1}+e_{1,\tau_{4}}^{k+1}-2\right) \\
y_{1,\tau_{3}}^{k}+\mu\left(-v_{\tau_{1}}^{k+1}-v_{\tau_{2}}^{k+1}+v_{\tau_{3}}^{k+1}\right)+e_{1,\tau_{3}}^{k+1}
\end{array}\right)\right].
\end{aligned}
\end{equation}
Since $\mathbf{w}=[0,0,0,2]^T$, $\mathbf{y}_{1,\tau}^{0,k+1}$ can be further derived as
\begin{equation}
\begin{aligned}
\mathbf{y}_{1,\tau}^{0^{k+1}} &=\left[\begin{array}{c}
y_{1,\tau_{2}}^{k}+\mu\left(\left(-v_{\tau_{1}}^{k+1}+v_{\tau_{2}}^{k+1}-v_{\tau_{3}}^{k+1}\right)+e_{1,\tau_{2}}^{k+1}-w_{2}\right) \\
y_{1,\tau_{1}}^{k}+\mu\left(\left(v_{\tau_{1}}^{k+1}-v_{\tau_{2}}^{k+1}-v_{\tau_{3}}^{k+1}\right)+e_{1,\tau_{1}}^{k+1}-w_{1}\right) \\
y_{1,\tau_{4}}^{k}+\mu\left(\left(v_{\tau_{1}}^{k+1}+v_{\tau_{2}}^{k+1}+v_{\tau_{3}}^{k+1}\right)+e_{1,\tau_{4}}^{k+1}-w_{4}\right) \\
y_{1,\tau_{3}}^{k}+\mu\left(\left(-v_{\tau_{1}}^{k+1}-v_{\tau_{2}}^{k+1}+v_{\tau_{3}}^{k+1}\right)+e_{1,\tau_{3}}^{k+1}-w_{3}\right)
\end{array}\right] \\
&=\left[y_{1,\tau_{2}}^{k+1} , y_{1,\tau_{1}}^{k+1}, y_{1,\tau_{4}}^{k+1} ,y_{1,\tau_{3}}^{k+1}\right]^T,
\end{aligned}
\end{equation}
i.e.,
\begin{equation}\label{y1tau0k+1=My1tauk+1}
\mathbf{y}_{1,\tau}^{0,k+1}=\mathcal{M}_{\mathbf{v}_\tau}(\mathbf{y}_{1,\tau}^{k+1}),
\end{equation}
when $\mathbf{v}_\tau=[1,1,0]^T$.

With similar derivation, we can see \eqref{y1tau0k+1=My1tauk+1} also holds when $\mathbf{v}_\tau=[1,0,1]^T, [0,1,1]^T, [0,0,0]^T$. Since $\mathbf{y}_{1}^{k+1}$ and $\mathbf{y}_{1}^{0,k+1}$ are cascaded by $\mathbf{y}_{1,\tau}^{k+1}$ and $\mathbf{y}_{1,\tau}^{0,k+1}$ respectively, $\tau=1,\cdots,\Gamma_c$, we can get
\begin{equation}\label{y10k+1=My1k+1}
\mathbf{y}_{1}^{0,k+1}=\mathcal{M}_{\mathbf{v}}(\mathbf{y}_{1}^{k+1}).
\end{equation}

Finally, it remains to verify that $\mathbf{y}_2^{0,k+1}=\mathcal{T}_{\mathbf{v}}(\mathbf{y}_2^{k+1})$. According to $\eqref{subproblems-e}$, when $v_i=1$, we have
\begin{equation}\label{y2i0k+1=-y2ik+1-detail}
\begin{aligned}
y_{2,i}^{0,k+1}&=y_{2,i}^{0,k}+\mu(v_i^{0,k+1}-e_{2,i}^{0,k+1})\\
&=-(y_{2,i}^{k}+\mu(v_i^{k+1}-e_{2,i}^{k+1})),
\end{aligned}
\end{equation}
where the second equality holds since $y_{2,i}^{0,k}=\mathcal{K}_{\mathbf{v}}(y_{2,i}^{k})=-y_{2,i}^{k}$, $v_i^{0,k+1}={\rm{R}}_{\mathbf{v}}(v_i^{k+1})=1-v_i^{k+1}$ and $e_{2,i}^{0,k+1}={\rm{R}}_{\mathbf{v}}(e_{2,i}^{k+1})=1-e_{2,i}^{k+1}$ when $v_i=1$. Thus from equation \eqref{y2i0k+1=-y2ik+1-detail}, when $v_i=1$, we have
\begin{equation} \label{y2i0k+1=-y2ik+1}
y_{2,i}^{0,k+1}=-y_{2,i}^{k+1}.
\end{equation}

For the other case when $v_i=0$ is transmitted over the channel, there exist
\begin{equation}\label{y2i0k+1=y2ik+1}
\begin{aligned}
y_{2,i}^{0,k+1}&=y_{2,i}^{0,k}+\mu(v_i^{0,k+1}-e_{2,i}^{0,k+1})\\
&=y_{2,i}^{k}+\mu(v_i^{k+1}-e_{2,i}^{k+1})\\
&=y_{2,i}^{k+1},
\end{aligned}
\end{equation}
where the second equality holds since $y_{2,i}^{0,k}=\mathcal{K}_{\mathbf{v}}(y_{2,i}^{k})=y_{2,i}^{k}$, $v_i^{0,k+1}={\rm{R}}_{\mathbf{v}}(v_i^{k+1})=v_i^{k+1}$ and $e_{2,i}^{0,k+1}={\rm{R}}_{\mathbf{v}}(e_{2,i}^{k+1})=e_{2,i}^{k+1}$ when $v_i=0$.

Combine \eqref{y2i0k+1=-y2ik+1} and \eqref{y2i0k+1=y2ik+1}, we get
\begin{equation}
\mathbf{y}_{2}^{0,k+1}=\mathcal{K}_{\mathbf{v}}(\mathbf{y}_{2}^{k+1}),
\end{equation}
This ends the proof of Lemma \ref{lemma-symmetry-v0k-v0k+1}.
\end{proof}


%





\ifCLASSOPTIONcaptionsoff
  \newpage
\fi

\end{document}